\documentclass[onecolumn]{IEEEtran}
\usepackage{amsmath,amsfonts,amsthm}
\usepackage{mathrsfs}
\usepackage{amssymb}
\usepackage{cite}

\newtheorem{theorem}{Theorem}
\newtheorem{lemma}[theorem]{Lemma}
\newtheorem{prop}[theorem]{Proposition}
\newtheorem{corollary}[theorem]{Corollary}

\newcounter{example}  
\newenvironment{example}[1][]{\refstepcounter{example}\par\medskip
   \noindent \textit{Example~\theexample. #1} \rmfamily}{\medskip}

\newcommand{\cC}{\mathcal{C}}
\newcommand{\cD}{\mathcal{D}}
\newcommand{\cR}{\mathcal{R}}

\newcommand{\cX}{\mathcal{X}}

\newcommand{\bc}{\mathbf{c}}
\newcommand{\bd}{\mathbf{d}}

\newcommand{\bu}{\mathbf{u}}
\newcommand{\bv}{\mathbf{v}}

\DeclareMathOperator{\supp}{supp}

\title{Optimal Quaternary $(r,\delta)$-Locally Repairable Codes Achieving the Singleton-type Bound}

\author{Kenneth W. Shum\thanks{Kenneth W. Shum is with School of Science and Engineering, The Chinese University of Hong Kong (Shenzhen), Shenzhen, 518055, China.} and Jie Hao
\thanks{Jie Hao is with the Information Security Center, Beijing University of Posts and Telecommunications, Beijing, 100876, China.}}

\begin{document}
\maketitle

\begin{abstract}
Locally repairable codes enables fast repair of node failure in a distributed storage system. The code symbols in a codeword are stored in different storage nodes, such that a disk failure can be recovered by accessing a small fraction of the storage nodes. The number of storage nodes that are contacted during the repair of a failed node is a parameter called locality.  We consider locally repairable codes that can be locally recovered in the presence of multiple node failures.  The punctured code obtained  by removing the code symbols in the complement of a repair group is called a local code. We aim at designing a code such that all local codes have a prescribed minimum distance, so that any node failure can be repaired locally, provided that the total number of node failures is less than the tolerance parameter. We consider linear locally repairable codes defined over a finite field of size four. This alphabet has characteristic 2, and hence is amenable to practical implementation. We classify all quaternary locally repairable codes that attain the Singleton-type upper bound for minimum distance. For each combination of achievable code parameters, an explicit code construction is given.

\end{abstract}


\section{Introduction}

The advent of locally repairable codes is motivated by their application in distributed storage systems. A data file is first divided into chunks of $k$ data symbols. Each chunk is encoded separately by an $[n,k]$ linear code. The $n$ encoded symbols are distributed and stored in $n$ different storage nodes. In this paper we will use the terms ``storage node'' and ``disk'' interchangeably. If a disk failure occurs, we replace the failed node by a new one, and recover the lost data with the help of the surviving storage nodes. A basic recovery method is to contact $k$ other storage nodes, decode the original data file, and then re-encode the lost data. This recovery method incurs excessive network traffic as whole data file is downloaded, but in fact only a small portion of the data file is required. 

The network traffic generated during the course of repairing a failed node is termed {\em repair bandwidth} in~\cite{Dimakis2010}. The authors in~\cite{Dimakis2010} suggest that the new node can contact and download data from all surviving nodes, and each of the other surviving nodes send minimum amount of coded data to the new node, so as to minimize the repair bandwidth. An information-theoretic lower bound on the resulting repair bandwidth is derived, and is shown to be much smaller than the size of the original data file.  This type of encoding schemes that minimize repair bandwidth is called {\em regenerating codes}. 

Another approach is to minimize the number of contacted nodes during the repair process. In this repair model, all coded symbols in the contacted node are sent to the new node. The number of contacted nodes is a code parameter called {\em locality}, which is commonly denoted by the letter~$r$. A code symbol is said to have locality $r$ if it is a deterministic function of $r$ other code symbols. A code with this kind of local property is called a {\em locally repairable code}, or {\em locally recoverable code} (LRC). There are two categories of LRCs. In the first category, we only require that there is a set of $k$ information symbols such that each of them has locality $r$; the parity-check symbols need not be locally recoverable. A code with this property is said to be a code with {\em information locality}. A fundamental bound on the LRC with information locality is derived in~\cite{Gopalan}. An implementation of this idea in Hadoop file system is investigated in~\cite{SAPDVCB13}. In the second category of LRCs, all coded symbols are required to have locality $r$. A code with this type of local property is said to have {\em all-symbol locality}. Obviously, the latter imposes a more stringent requirement on the code design in compare to the former, and the bound in~\cite{Gopalan} is also valid for LRCs with all-symbol locality. In this paper, we will consider LRCs with all-symbol locality. There are many results in the literature about LRCs. We refer the readers to~\cite{Siberstein13, TamoBarg14 , Song14, Papailiopoulos14, CadambeMazumdar15, HYUS16, GXY19} for a partial list.

In~\cite{PKLK12}, the repair model mentioned in the previous paragraph is extended to enable local repair under multiple node failures. 
The authors in~\cite{PKLK12} give an upper bound on the minimum distance that extends the classical Singleton bound to codes with local properties. A refinement of this bound is recently derived in~\cite{Cai2021}.  There are a number of constructions of optimal code in this multiple-failure setting. In the literature, we can find constructions using cyclic or constacyclic codes~\cite{CXHF18,CFXF19,CH19, QZF21}, generalized RS codes~\cite{Zhang2020}, and Steiner system~\cite{CS21}, and finite geometry\cite{FLY}, etc. Another source of code constructions is algebraic geometry. Constructions of LRC using algebraic curves and algebraic surfaces are studied in~\cite{Barg17, LMX19,Salgado21}. In this paper, we will focus on optimal linear codes that attain the Singleton-like bound in the multiple-failure setting.

Code design for practical distributed storage systems is limited by computational complexity. The speed of encoding and decoding algorithms is affected by the choice of alphabet. Many existing code constructions of LRC, including those that are based on Reed-Solomon codes and cyclic codes, require the field size to be larger than the code length. For example, families of LRCs whose lengths grow super-linearly as a function of the alphabet size are constructed in~\cite{Cai2020,KWG21}. We take a slightly different perspective and enumerate the code parameters of all optimal LRCs over a fixed finite field. For field size two and three, binary and ternary locally repairable codes attaining the Singleton-like bound are already characterized in~\cite{HaoICC2017, Hao2017, Hao2019, Hao2020}. In this paper, we classify optimal LRCs over $GF(4)$ attaining the Singleton-like bound for LRCs that can locally correct multiple node erasures. A complete enumeration of all possible code parameters is obtained (see Table~\ref{table:code_parameters}). 

An explicit code construction is given to each combination of available code parameters.   Some of the construction is presented using a parity-check matrix which has the structure of {\em generalized tensor product} (GTP).
Code constructions using GTP was proposed by Imai and Fujiya in~\cite{IF81}, which extended the result in~\cite{Wolf65}. Construction of LRC using generalized tensor product is explored in~\cite{HYS20}. However, extension fields of $GF(q)$ are considered in~\cite{IF81} and \cite{HYS20}. In this paper we present a simplified  version if it as field extension is not required. The structure of the code is equivalent to a {\em generalized integrated interleaved code} \cite{BlaumHetzler18, Blaum20, LLW20}, and {\em matrix-product code}~\cite{BlackmoreNorton, OzbudakStichtenoth,LEL2021}.  We remark that the GTP construction is for LRC with disjoint repair groups. In this paper we do not assume that the repair groups are disjoint. Indeed, some of the optimal codes belong to the families of simplex codes and ovoid codes, and they have overlapping repair groups

This paper is organized as follows. In Section~\ref{sec:preliminaries} we review some basic results of LRC. In Section~\ref{sec:GTP} we give a simplified version of generalized tensor product, which will be used for code constructions in this paper. In Section~\ref{sec:r=1}, we classify the optimal quaternary LRCs with locality $r=1$. In Sections~\ref{sec:d=3} and \ref{sec:d=4}, we enumerate the optimal quaternary LRCs with distance 3 and~4, respectively. For distance larger than 4, we split the presentation into Sections~\ref{sec:d>4a} and \ref{sec:d>4b}; Section~\ref{sec:d>4a} is for low rate and Section~\ref{sec:d>4b} is for higher rate. We conclude the paper with some remarks in~Section~\ref{sec:conclusion}.

\begin{table*}
\caption{Parameters of Singleton-optimal $(r,\delta)$-LRC over $GF(4)$ for $\delta>2$.}
\label{table:code_parameters}
$$
\begin{array}{|c|c|c|c|c|c|c|} \hline
 r & \delta & n & k & d & \text{Remark} & \text{Reference}\\ \hline \hline
 1 & \geq 3     & k\delta & k& \delta & k\geq 2& \eqref{LRC:r=1_d=delta} \\ \hline
 1 & \geq 3  & (k+1)\delta & k & 2\delta &k\geq 2 & \eqref{LRC:r=1_d=2delta} \\ \hline
 1 & \geq 3 & (k+2)\delta  & k & 3\delta & 2,3& \eqref{LRC:r=1_d=3delta} \\ \hline
 1 & \geq 3 & (k+3)\delta & k & 4\delta & 2,3 & \eqref{LRC:r=1_d=4delta} \\ \hline \hline
 2 & 3 & 4\ell-e & 2\ell-e & 3 & \ell\geq 2,\ e=0,1&\eqref{LRC:r=2_delta=3_d=3},\eqref{LRC:r=2_k=1_delta=3_d=3}\\ \hline
 3 & 3 & 5\ell-e & 3\ell-e & 3 & \ell\geq 2, \  e=0,1,2& \eqref{LRC:r=3_delta=3_d=3}, \eqref{LRC:r=3_k=2_delta=3_d=3}, \eqref{LRC:r=3_k=1_delta=3_d=3}\\ \hline \hline
 2 & 3 & 4\ell & 2\ell-1 & 4 & \ell\geq 2 & \eqref{LRC:r=2_k=1_delta=3_d=4} \\ \hline
 3 & 3 & 5\ell-e & 3\ell-1-e & 4 & \ell \geq 2,\ e=0,1 & \eqref{LRC:r=3_k=1_delta=3_d=4}, \eqref{LRC:r=3_k=1_delta=3_d=4b} \\ \hline
 2 & 4 & 5\ell-e&2\ell-e& 4 & \ell\geq 2,\ e=0,1 & \eqref{LRC:r=2_delta=4_d=4} \\ \hline
 3 & 4 & 6\ell-e&3\ell-e& 4 & \ell\geq 2,\ e=0,1,2 &\eqref{LRC:r=3_delta=4_d=4} \\ \hline \hline
2 & 4 & d+5 & 3 & d & 7\leq d\leq 16 &\eqref{LRC:r=2_delta=4_d>4} \\ \hline
3 & 3 & d+4 & 3 & d & 5\leq d\leq 12 &\eqref{LRC:r=2_delta=3_d>4} \\ \hline
3 & 3 & d+5 & 4 & d & 5\leq d\leq 12 &\eqref{LRC:r=3_delta=3_d>4} \\ \hline
3 & 4 & d+6 & 4 & d & 6 \leq d \leq 12 &\eqref{LRC:r=3_delta=4_d>4} \\ \hline \hline
2 & 3 & 4\ell & 2\ell-3 & 8 & \ell=3,4,5 & \eqref{LRC:r=2_delta=3_large1}\\
\hline
2 & 3 & 4\ell & 2\ell-5 & 12 & \ell=4,5 & \eqref{LRC:r=2_delta=3_large2}\\
\hline
3 & 3 & 5\ell & 3\ell-2 & 5 & \ell\geq 3& \eqref{LRC:r=3_delta=3_large1}\\ \hline
3 & 4 & 6\ell & 3\ell-2 & 6 & \ell \geq 3 & \eqref{LRC:r=3_delta=4_large1} \\ \hline
3 & 4 & 5\ell & 3\ell-5 & 12& 3\leq\ell\leq 17 & \eqref{LRC:r=3_delta=4_large2}\\ \hline
\end{array}
$$
\end{table*}

\section{Preliminaries on Optimal Locally Repairable Codes}
\label{sec:preliminaries}

\subsection{Definition of Locally Repairable Codes}

We give a formal definition of linear locally repairable codes over a general finite field $GF(q)$, where $q$ is a prime power.  A linear code $\cC$ over $GF(q)$ with length $n$, dimension $k$ and minimum distance $d$ is called an $[n,k,d]_q$ code, or $[n,k,d]$ code if the alphabet size $q$ is understood from the context. The dimension and minimum distance of $\cC$ are denoted by $\dim(\cC)$ and $d_{\mathrm{min}}(\cC)$, respectively. 
We label the $n$ components of a codeword in $\cC$ by $[n]:=\{1,2,\ldots, n\}$. Given a subset $\cX$ of $[n]$, we let $\cC_\cX$ denote the {\em punctured subcode}
$$
 \cC_\cX := \{ (u_j)_{j\in \cX} \in GF(q)^{|\cX|} :\, (u_1,\ldots, u_n)\in \cC \},
$$
which is also called the {\em restriction} of $\cC$ at $\cX$. We will assume that $\cC_{\mathcal{S}} \neq \{\mathbf{0}\}$ for any subset $\mathcal{S}$ of $\{1,2,\ldots, n\}$, i.e., no component is identically zero.

For $r\geq 1$ and $\delta \geq 2$, we say that $\cC$ has {\em locality} $r$ with {\em tolerance} $\delta$ if $\cC$ there is a collection $\mathscr{R}$ of subsets of $[n]$ such that

(a) $[n] = \bigcup_{\cR \in \mathscr{R}} \cR$,

(b) $|\cR| \leq r+ \delta-1$ for all $\cR\in\mathscr{R}$,

(c) the minimum distance of $\cC_{\cR}$ is at least~$\delta$ for all $\cR\in\mathscr{R}$.

Each $\cR\in\mathscr{R}$ is called a {\em repair group}, and the punctured subcode $\cC_{\cR}$ is called a {\em local code}. By a slight abuse of notation, we will also call the code symbols indexed by some $\cR\in\mathscr{R}$ a repair group. Condition (c) forces the size of a repair group to be at least $\delta$. Conditions (b) and (c) guarantee that a code symbol in a repair group $\cR$ can be recovered by accessing any $|\cR|-(\delta-1) \leq r$ other code symbols in the repair group $\cR$. Since the repair groups cover all the coded symbols by assumption (a), we can recover any $\delta$ erasures locally. A code $\cC$ that satisfies the above requirements is denoted by $(r,\delta)$-LRC. We note that an $(r,\delta)$-LRC is also an $(r',\delta')$-LRC for all $r'\geq r$ and $2\leq \delta'\leq \delta$.

For $r\geq 1$ and $\delta \geq 2$, the minimal distance of an $(r,\delta)$-LRC with length $n$ and dimension $k$ satisfies a Singleton-type bound \cite{PKLK12},
\begin{equation}
d \leq n -k + 1 - \Big( \Big\lceil \frac{k}{r} \Big\rceil -1 \Big)(\delta - 1).
\label{eq:Singleton_like}
\end{equation}
An $(r,\delta)$-LRC achieving this bound with equality is called {\em Singleton-optimal}, or {\em optimal} for short. 

When $k\leq r$, the bound in \eqref{eq:Singleton_like} reduces to the classical Singleton bound, and an optimal LRC attaining~\eqref{eq:Singleton_like} is an MDS code. We will assume $r<k$ in the rest of this paper.

We recall how to define locality in terms of the {\em dual code}
$$
 \cC^\perp := \{  \bu \in GF(q)^n:\, \bu \cdot \bv = 0 \text{ for all } \bv \in \cC \}.
$$
We denote the {\em support} of a linear code $\cD$ of length $n$ by
$$
\supp(\cD) := \{ i\in [n]:\,   \text{there exists } \bu \in \cD \text{ s.t. }u_i \neq 0\}.
$$
A linear code $\cC$ has locality $r$ with tolerance $\delta$ if and only if  we can find a collection $\mathscr{D}$ of subcodes in $\cC^\perp$ such that

(a') $[n] = \bigcup_{\cD\in\mathscr{D}} \supp(\cD)$,

(b') $| \supp(\cD) | \leq r+\delta - 1$ for all $\cD\in\mathscr{D}$,

(c') For any $\cD\in\mathscr{D}$ and for any subset $\cX$ of $\supp(\cD)$ of size less than or equal to $\delta-1$, the dimension of $(\cD)_\cX$ is equal to $|\cX|$. 

Because we assume that no code symbol is identically zero, a subcode $\cD$ of $\cC^\perp$ satisfying Condition (c') must have support size larger than or equal to $\delta$. Condition (c') ensures that $\cD$ has dimension at least $\delta-1$. 
It is not difficult to verify that the support of a subcode $\cD$ that satisfies Conditions (b') and (c') is a repair group of~$\cC$.
  A generator matrix of a subcode $\cD \subset \cC^\perp$ satisfying Conditions (b') and (c') is called a {\em local parity-check matrix}. 

We remark that any $\cD$ in $\mathscr{D}$ is a subcode of $\cC^\perp$ with length $n$, and the punctured code $(\cD)_{\supp(\cD)}$ has length at most $r+\delta - 1$.

\begin{example} Consider the $[9,5,3]$ code over $GF(4)$ defined by the following parity-check matrix
$$
H=\left[
\begin{array}{ccccccccc}
        1&1&1&1&0&0&0&0&0 \\
        0&1&\alpha&\alpha^2&1&0&0&0&0\\
        0&0&0&0&0&1&1&1&1\\
        0&0&0&0&1&1&\alpha&\alpha^2&0
	\end{array} \right]
$$
where $\alpha$ is a primitive cube root of unity in $GF(4)$. The first two rows of $H$ form a local parity-check matrix, and the third and fourth rows form another local parity-check matrix. The supports of the two repair groups are $\{1,2,3,4,5\}$ and $\{5,6,7,8,9\}$, respectively. The two repair groups cover all the code symbols and overlap at the 5th component. This parity-check matrix defines an optimal $(3,3)$-LRC with minimum distance~3.
\label{ex:1}
\end{example}

\begin{example} The parity-check matrix in Fig.~\ref{fig:18_4_12} defines an optimal $(3,4)$-LRC with minimum distance~12. The symbols $\alpha$ and $\beta$ represent the roots of $x^2+x+1$ in $GF(4)$. There are three repair groups in this LRC and each of them has size~6. The six rows at the bottom of $H$ correspond to global parity-check equations.
\label{ex:2}
\end{example}

\begin{figure}
$$H={\small \left[
\begin{array}{cccccc|cccccc|cccccc}
1&0&0&1&1&1&          0&0&0&0&0&0& 0&0&0&0&0&0 \\
0&1&0&1&\alpha&\beta& 0&0&0&0&0&0& 0&0&0&0&0&0 \\
0&0&1&1&\beta&\alpha& 0&0&0&0&0&0& 0&0&0&0&0&0 \\ \hline 
0&0&0&0&0&0&          1&0&0&1&1&1& 0&0&0&0&0&0  \\
0&0&0&0&0&0&          0&1&0&1&\alpha&\beta& 0&0&0&0&0&0\\
0&0&0&0&0&0&          0&0&1&1&\beta&\alpha& 0&0&0&0&0&0 \\ \hline
0&0&0&0&0&0&          0&0&0&0&0&0& 1&0&0&1&1&1 \\
0&0&0&0&0&0&          0&0&0&0&0&0&0&1&0&1&\alpha&\beta \\
0&0&0&0&0&0&          0&0&0&0&0&0&0&0&1&1&\beta&\alpha \\ \hline
1&1&\alpha&\beta&0&\beta& 1&\alpha&\alpha&\beta&0&\beta&1&\beta&\alpha&1&0&\beta \\
\beta&\alpha&\alpha&\alpha&\alpha&\alpha& 1&0&1&\beta&\alpha&0&  \alpha&1&1&\beta&\beta&1 \\
0&1&1&\beta&\beta&1& \alpha&\beta&\alpha&1&\beta&1&  \alpha&\beta&\alpha&\alpha&1&b \\
1&\beta&\beta&\beta&\alpha&0& \beta&\alpha&1&\alpha&0&\beta&  \alpha&1&1&\beta&1&\alpha \\
1&\beta&\beta&\beta&\beta&1& \beta&\alpha&\alpha&1&1&\beta&  \beta&\alpha&1&\alpha&\alpha&\alpha 	
\end{array}
\right]}
$$
\caption{Parity-check matrix of a $(3,4)$-LRC with length $n=18$, dimension $k=4$ and minimum distance $d=12$.}
\label{fig:18_4_12}
\end{figure}

\subsection{Chain of Subcodes and Residue Codes}

\noindent {\bf Definition.} Let $\mathscr{D}$ denote the collection of all subcodes of $\cC^\perp$ that satisfies Conditions (b') and (c'). A sequence of subcodes $\cD_1$, $\cD_2,\ldots, \cD_m$ in $\mathscr{D}$ is said to form a {\em chain} if 
$$
\dim(\cD_1\oplus\cD_2\oplus\cdots\oplus\cD_i) - 
\dim(\cD_1\oplus\cD_2\oplus\cdots\oplus\cD_{i-1}) \geq \delta - 1
$$
for $i=2,3,\ldots, m$. The number of subcodes in a chain is called the {\em length} of the chain. A chain $\cD_1$, $\cD_2,\ldots, \cD_M$ is said to be {\em maximal} if 
$$
\dim(\cD_1\oplus\cdots\oplus\cD_M \oplus \cD') - 
\dim(\cD_1\oplus\cdots\oplus\cD_M) < \delta - 1
$$
for any $\cD' \in\mathscr{D}$. (The notation $\mathcal{D} \oplus \mathcal{D}'$ stands for the sum space of subspaces $\mathcal{D}$ and $\mathcal{D}'$.)

A lower bound on the length of a maximal chain is derived in~\cite{HaoICC2017, Hao2019}.

\begin{prop} Given an $(r,\delta)$-LRC $\cC$ with dimension $k$, a maximal chain of subcodes in $\cC^\perp$ contains at least $\lceil k/r \rceil$ subcodes.
\label{prop:chain}
\end{prop}



The next definition assumes that $\cC$ is an optimal LRC satisfying the Singleton-like bound~\eqref{eq:Singleton_like} with equality. 

\noindent {\bf Definition.} Suppose $\cC$ is a Singleton-optimal $(r,\delta)$-LRC. Let $\cD_1,\ldots, \cD_m$ be a chain with size $m=\lfloor (k-1)/r \rfloor$ (such a chain exists by Prop.~\ref{prop:chain}). The shortened code of $\cC$ consisting of codewords that are zero in the support of $\cD_1\oplus\cdots\oplus\cD_m$ is called a {\em residue code} of~$\cC$.

The following result is proved in~\cite[Theorem 1]{Hao2019}.

\begin{prop}
Suppose $\mathcal{C}$ is a Singleton-optimal $(r,\delta)$-LRC with dimension $k$. \begin{enumerate}
\item[(i)] A residue code is an MDS code with minimum distance $n-k-\lfloor(k-1)/r \rfloor(\delta-1)+1$ and co-dimension $n-k-\lfloor(k-1)/r \rfloor(\delta-1)$.

\item[(ii)] Any subcode in $\cC^\perp$ that satisfies Conditions (b') and (c'), i.e., a subcode in $\cC^\perp$ that corresponds to a repair group of $\cC$, has dimension $\delta-1$.

\item[(iii)] The restriction of $\cC$ to any repair group is an MDS code with minimum distance exactly $\delta-1$.
\end{enumerate}
\label{prop:residue_code}
\end{prop}


Part (iii) in the above proposition is proved in \cite[Theorem 2.2]{Kamath14} under an additional condition that $r|k$. This result actually holds holds without the condition $r|k$. For the sake of completeness, we give the proofs of Prop.~\ref{prop:chain} and Prop.~\ref{prop:residue_code} in Appendix~\ref{app:A}.

We note that there are more than one residue codes; two different chains can give rise to two residue codes with different supports. Two residue codes may have different support length. For example, in Example~\ref{ex:1}, the residue code of the LRC is a $[4,2,3]$ codes over $GF(4)$ defined by the parity-check matrix
$$
\begin{bmatrix}
1 & 1 & 1 & 1 \\
1& \alpha&\alpha^2& 0
\end{bmatrix}.
$$
In Example~\ref{ex:2}, the residue code is a $[12,1,12]$ repetition code.

The next proposition investigates the possible dimension of a residue code.

\begin{prop}
Suppose $\cC$ is an optimal $(r,\delta)$-LRC with dimension $k$. We pick a maximal chain of subcodes $\cD_1,\ldots, \cD_M$ in $\cC^\perp$ and construct the associated residue code. The dimension $k'$ of the residue codes satisfies 
\begin{equation}
k' \geq k - \big\lfloor \frac{k-1}{r} \big\rfloor r.
\label{eq:range_k_r}
\end{equation}
Moreover, if equality holds, the supports of the first $\lfloor (k-1)/r \rfloor$ subcodes in a chain are mutually disjoint.
\label{prop:range_k_r}
\end{prop}

\begin{proof}
 Suppose that the residue code of is an $[n',k']$ code. From Prop.~\ref{prop:residue_code}, we know that the co-dimension equals
$$
n'-k' = d-1 = n-k-\big\lfloor \frac{k-1}{r} \big\rfloor(\delta-1).
$$
After some re-arrangements, this becomes
$$
n-n' = k + \big\lfloor \frac{k-1}{r} \big\rfloor(\delta-1)-k'.
$$
Since $n-n'$ is the size of the union of the repair groups associated with the  first $\lfloor (k-1)/r \rfloor$ subcodes in the chain, we obtain
\begin{equation}
 k + \big\lfloor \frac{k-1}{r} \big\rfloor(\delta-1)-k' \leq 
\big\lfloor \frac{k-1}{r} \big\rfloor(r+\delta-1)
\label{eq:range_k_r_proof}
\end{equation}
which is equivalent to the inequality in the proposition.

If equality holds in \eqref{eq:range_k_r_proof} holds, then the support of $\cD_1,\ldots, \cD_{\lfloor (k-1)/r \rfloor}$ must be mutually disjoint.
\end{proof}


\smallskip

Prop.~\ref{prop:chain}  suggests that we can write down a parity-check matrix of an $(r,\delta)$-LRC in a specific form. Given an $(r,\delta)$-LRC $\cC$, we first pick a maximal chain of subcodes in $\cC^\perp$. We form a parity-check matrix by stacking the local parity-check matrices associated to the subcodes in the maximal chain. If the number of rows in this matrix is strictly less than $n-k$, we append additional rows to this matrix to form a parity-check matrix of $\cC$. The structure of the resulting parity-check matrix is summarized as follows.

\begin{corollary}
For any Singleton-optimal $(r,\delta)$-LRC with length $n$ and dimension $k$, there is a parity-check matrix with the following structure:
\begin{itemize}
\item The rows are divided into an upper part and a lower part;
\item the rows in the upper part of the matrix are partitioned into groups, with each group corresponding to a repair group;
\item there are at least $\lceil k/r \rceil$ groups, and each of these group occupies exactly $\delta-1$ rows;
\item any $\delta-1$ columns chosen within the support of a repair group are linearly independent;
\item the lower part of the matrix, which may or may not be empty, corresponds to global parity-check equations of the code.
\end{itemize}
\label{cor:H}
\end{corollary}

An example of parity-check matrix that is in the form described by the above corollary can be found in Fig.~\ref{fig:18_4_12}. In the matrix in Fig.~\ref{fig:18_4_12}, the first 9 rows are divided into three groups, and each group corresponds to a local parity-check matrix. The last five rows are global parity-check equations.

In general, we can put the parity-check matrix in the form
\begin{equation}
H = \left[
\begin{array}{c|c}
A & 0 \\ \hline
B & C
\end{array} \right],
\label{eq:parity_check_matrix}
\end{equation}
where $A$ is a submatrix containing $\lfloor (k-1)/r \rfloor(\delta-1)$ rows, and $C$ is a parity-check matrix of the residue code.


\medskip

We will use some result for almost MDS code  \cite{Boer96} \cite{DL95}, whose definition is reviewed below. Given a $q$-ary $[n,k,d]$ linear code~$\cD$, we define the {\em Singleton defect}, or simply {\em defect}, of $\cD$ by
$$
s(\cD) := n-k+1-d.
$$
The defect of $\cD$ is a non-negative integer, and is equal to zero if and only if $\cD$ is MDS. From the Singleton-type bound~\ref{eq:Singleton_like}, we see that the Singleton defect of a $k$-dimensional $(r,\delta)$-LRC must be larger than or equal to $\lfloor (k-1)/r \rfloor(\delta-1)$.

We record a useful bound on the minimum distance in the next theorem.

\begin{theorem}[\cite{FW97}]
The minimum distance of a $q$-ary linear code with dimension $k\geq 2$ and defect $s$ is less than or equal to $q(s+1)$.
\label{thm:defect}
\end{theorem}




Using Theorem~\ref{thm:defect}, we can prove the following theorem.

\begin{theorem} If $\cC$ is a $q$-ary $k$-dimensional Singleton-optimal $(r,\delta)$-LRC with minimum distance $d$, then
$$
 d \leq \begin{cases}
q & \text{ if } r \nmid (k-1), \\
\delta q & \text{ if } r | (k-1).
\end{cases}
$$
\label{thm:d}
\end{theorem}

\begin{proof}
See \cite[Theorem 1]{Hao2019}.
\end{proof}

The next example shows that the bound $d\leq \delta q$ when $r|(k-1)$ is tight.

\begin{example}
Consider an MDS code $\cC_2$ over $GF(q)$ of dimension $k=2$ and length $q+1$. For example, we can take the extended Reed-Solomon code with generator matrix
$$
\begin{bmatrix} 0&1&1&\cdots&1 \\ 1&\beta_1&\beta_2&\cdots&\beta_q \end{bmatrix},
$$
where $\beta_1,\ldots, \beta_q$ are distinct elements in $GF(q)$. We concatenate this code with a repetition code of length $\delta>2$ as the inner code. The resulting code has length $\delta(q+1)$. For each codeword $(c_1,c_2,\ldots, c_{q+1})$ in $\cC_2$, the codeword after concatenation has the form
$$
(\underbrace{c_1,\ldots,c_1}_{\delta \text{ times}}, 
\underbrace{c_2,\ldots,c_2}_{\delta \text{ times}},
\ldots,
\underbrace{c_{q+1},\ldots,c_{q+1}}_{\delta \text{ times}}).
$$
It is obvious that the concatenated code has minimum distance $\delta q$ and dimension $k=2$. It is an optimal $(1,\delta)$-LRC, attaining the upper bound in Theorem~\ref{thm:d}.
\label{example_concatenate}
\end{example}

\section{Code Construction Using Generalized Tensor Product}
\label{sec:GTP}

In the following, we present a simplified version of generalized tensor product code that does not require extension field. 

\smallskip

{\bf Construction 1.} (GTP construction for $(r,\delta)$-LRC) Suppose that $\ell$, $\mu$, and $\nu$ are positive integers satisfying $\mu < \ell$ and $\nu\leq r$.
Let $I_{\ell}$ be the identity matrix of size $\ell\times \ell$. Given a $(\delta-1)\times (r+\delta-1)$ matrix $B_1$,  a $\mu\times \ell$ matrix $A_2$, and a $\nu \times (r+\delta-1)$ matrix $B_2$, all over $GF(q)$, define a linear code over $GF(q)$ of length $\ell (r+\delta-1)$ by the parity-check matrix
$$
H = \begin{bmatrix}
I_{\ell} \otimes B_1 \\
A_2 \otimes B_2
\end{bmatrix},
$$
where $\otimes$ denotes the tensor product operator for matrices. We denote the resulting code by $\cC(\ell, B_1, A_2, B_2)$.

The upper part of $H$ is a block diagonal matrix. Each non-zero block in the diagonal corresponds to a repair group of size $r+\delta-1$.


\begin{prop}
Suppose that $A_2$ is a parity-check matrix of an $[\ell, \ell-\mu]$ MDS code,  $B_1$ is a parity-check matrix of an $[r+\delta-1, r]$ MDS code, and
$
\begin{bmatrix} B_1 \\ B_2 \end{bmatrix}
$  
is a parity-check matrix of an $[r+\delta-1, r-\nu]$ MDS code, then the code $\cC(\ell,B_1,A_2,B_2)$ is an $(r,\delta)$-LRC code with dimension $\ell r  - \mu \nu$. The minimum distance is lower bounded by
$$
d_{\mathrm{min}}(\cC(\ell,B_1,A_2,B_2)) \geq
\begin{cases}
\delta(\mu+1)  & \text{if } r=\nu, \\
\min\{\delta(\mu+1), \delta+\nu\} & \text{if } r > \nu.
\end{cases}
$$ 
\label{prop:GTP}
\end{prop}

A proof of Prop.~\ref{prop:GTP} is given in the appendix.


\section{Optimal Quaternary $(1,\delta)$-LRC}
\label{sec:r=1}

 In the remaining of this paper we consider $(r,\delta)$-LRCs over $GF(4)$ that achieve the Singleton-type bound. We will assume $\delta>2$. The symbol $\alpha$ denotes a primitive cube root of unity in $GF(4)$. The elements in $GF(4)$ are represented by 0, 1, $\alpha$ and $\alpha^2$. Sometime, we will use the symbol $\beta$ to denote $\alpha^2$.

As the residue code must be an MDS code over $GF(4)$ with distance larger than~2, the only possible code parameters are

\begin{itemize}
\item (distance 3) $[5,3,3]$, $[4,2,3]$, $[3,1,3]$;
\item (distance 4) $[6,3,4]$, $[5,2,4]$, $[4,1,4]$;
\item (distance $>4$) $[n,1,n]$  for $n>4$.
\end{itemize}

When $r=1$, we have $r|k$ for any $k$. By \cite[Theorem 2.2]{Kamath14}, the repair groups are mutually disjoint and each of them has size $\delta$, corresponding to a repetition code of length $\delta$.  The block length $n$ is thus a multiple of $\delta$. Let $\ell$ be the positive integer $\ell=n/\delta$, so that the code symbols are covered by $\ell$ mutually disjoint repair groups. 

 The minimum distance is given by
$$
d = \ell \delta - k - (k-1)(\delta-1) + 1 = (\ell -k +1) \delta.
$$
Since $d$ must be positive, we get $\ell \geq k$.
By Theorem~\ref{thm:d}, we have $d=(\ell-k+1)\delta\leq 4\delta$. Hence,
 $\ell - k \leq  3$. From Prop.~\ref{prop:residue_code}, the co-dimension of a residue code is  $\ell \delta - k - (k-1)(\delta-1) = d-1$. The residue code is equivalent to a repetition code of length $d=(\ell-k+1)\delta$ when $r=1$.

In the following, we divide the discussion into three cases: $\ell=k,k+1$, $\ell=k+2$, and $\ell=k+3$.  We use the notation 
\begin{equation}
H_\delta \triangleq \left[
\begin{array}{cccc|c}
1&0&\cdots &0 & 1 \\
0&1& \cdots&0 & 1 \\
\vdots&\vdots&\ddots&\vdots&1\\
0&0&\cdots&1&1
\end{array}
\right]
\label{eq:H_delta}
\end{equation}
to denote the parity-check matrix of a $[\delta,1,\delta]$ repetition code.

\noindent \underline{\em Case $\ell=k, k+1$}. When $n=k\delta$, the code parameters are
\begin{equation}
n=k \delta,\  k\geq 2, \ r=1, \ \delta > 2, \ d=\delta.
\label{LRC:r=1_d=delta}
\end{equation}
An optimal LRC with these code parameters is a concatenation of $k$ repetition codes of length $\delta$. A parity-check matrix of $\cC$ is given by
$ I_k \otimes H_\delta$.

When $n=(k+1)\delta$, there are $k+1$ repair groups and each local code is a repetition code. The residue code is a repetition code with length $2\delta$.  The parameters of an optimal code in this case are
\begin{equation}
n=(k+1)\delta, \ k \geq 2, \ r=1,\ \delta > 2,\ d=2\delta.
\label{LRC:r=1_d=2delta}
\end{equation}
LRCs with the above parameters can be constructed by generalized tensor product $\cC(k+1,H_\delta,A_2,B_2)$ with $A_2=(1,1,\ldots, 1)$ and $B_2=(0,0,\ldots, 0,1)$. The minimum distance is $2\delta$ by Prop.~\ref{prop:GTP}.

For example, the following is a parity-check matrix of an optimal $(1,4)$-LRC of length 12, dimension 2 and minimum distance 8,
$$H= {\small \left[
\begin{array}{cccc|cccc|cccc}
1&1&1&1&  0&0&0&0& 0&0&0&0\\
0&1&\alpha&\alpha^2& 0&0&0&0& 0&0&0&0 \\
0&1&\alpha^2&\alpha& 0&0&0&0& 0&0&0&0 \\ \hline 
0&0&0&0& 1&1&1&1& 0&0&0&0\\
0&0&0&0& 0&1&\alpha&\alpha^2& 0&0&0&0 \\
0&0&0&0& 0&1&\alpha^2&\alpha& 0&0&0&0\\ \hline 
0&0&0&0& 0&0&0&0& 1&1&1&1\\
0&0&0&0& 0&0&0&0& 0&1&\alpha&\alpha^2  \\
0&0&0&0& 0&0&0&0 & 0&1&\alpha^2&\alpha \\ \hline
0&1&1&1& 0&1&1&1& 0&1&1&1  
\end{array} \right]}.
$$
A generator matrix for this code is
$$ G=
\left[
\begin{array}{cccc|cccc|cccc}
1&1&1&1&0&0&0&0&1&1&1&1 \\
0&0&0&0&1&1&1&1&1&1&1&1
\end{array}
\right].
$$
After some column permutations, we can transform the generator matrix $G$ in a tensor-product format,
$$
\begin{bmatrix}
0&1&1 \\
1&0&1
\end{bmatrix} \otimes 
\begin{bmatrix} 1&1&1&1 \end{bmatrix}.
$$

\smallskip

\noindent \underline{\em Case $\ell=k+2$}. When $\ell=k+2$, the residue code of an optimal LRC has minimum distance $3\delta$, and is equivalent to an $[3\delta, 1, 3\delta]$ repetition code. A parity-check matrix has the form
\begin{equation}
\left[
\begin{array}{cc|ccc}
H_\delta  &&&&  \\
& \ddots &&& \\ \hline
&&H_\delta  &&  \\ 
&&& H_\delta & \\
&&&& H_\delta \\
A_1 & \cdots & A_{\ell-2} &A_{\ell-1} & A_\ell
\end{array}
\right]
\label{eq:H_r=1_a}
\end{equation}
where $A_1,\ldots, A_\ell$ are matrices with dimension $2\times \ell$. The submatrix on the lower right corner is a parity-check matrix of the residue code of the LRC. Since $H_\delta$ is given as in \eqref{eq:H_delta}, we can perform row reductions to the above parity-check matrix and assume without generality that the first $\delta-1$ columns in $A_i$ are all zeros, for $i=1,2,\ldots, k+2$,
$$
A_i = \begin{bmatrix}
0&0&\cdots &0 & a_{i1} \\
0&0&\cdots &0 & a_{i2}
\end{bmatrix}.
$$

In order to ensure that the distance strictly is larger than $2\delta$, the column vectors $(a_{i1},a_{i2})^T$, for $i=1,2,\ldots, k+2$, must be {\em projectively distinct};  the one-dimensional vector spaces
$$
V_i := \{ (\lambda a_{i1},\ \lambda a_{i2}):\, \lambda\in GF(4)\}
$$
are distinct points in the projective space $PG_2(4)$. Since we have exactly five points in the projective space $PG_2(4)$, there are at most 5 repair groups in an optimal LRC with parameters $r=1$ and $n=(k+2)\delta$. A parity-check matrix of an optimal LRC  with 5 repair groups can be obtained by setting $\ell=5$ and
\begin{gather*}
A_1 = \begin{bmatrix}
0&0&\cdots &0 & 1 \\
0&0&\cdots &0 & 0
\end{bmatrix},\ 
A_2 = \begin{bmatrix}
0&0&\cdots &0 & 1 \\
0&0&\cdots &0 & 1
\end{bmatrix},\ 
A_3 = \begin{bmatrix}
0&0&\cdots &0 & 1 \\
0&0&\cdots &0 & \alpha 
\end{bmatrix}, \\ 
A_4 = \begin{bmatrix}
0&0&\cdots &0 & 1 \\
0&0&\cdots &0 & \beta
\end{bmatrix},\ 
A_5 = \begin{bmatrix}
0&0&\cdots &0 & 0 \\
0&0&\cdots &0 & 1 
\end{bmatrix}
\end{gather*}
in~\eqref{eq:H_r=1_a}.
The parameters of optimal LRC in this case are
\begin{equation}
n=(k+2)\delta, \ k =2,3, \ r=1,\ \delta > 2,\ d=3\delta.
\label{LRC:r=1_d=3delta}
\end{equation}

\smallskip

\noindent \underline{\em Case $\ell=k+3$}. The residue code of an optimal LRC in this case has minimum distance $4\delta$, and is equivalent to a $[4\delta, 1, 4\delta]$ repetition code. A parity-check matrix has similar structure as in~\eqref{eq:H_r=1_a}, except that the submatrices $A_i$ at the bottom are $3\times \delta$ matrices in the form
$$
\begin{bmatrix}
0&0&\cdots &0&a_{i1} \\
0&0&\cdots &0&a_{i2} \\
0&0&\cdots &0&a_{i3} 
\end{bmatrix}.
$$

In order to rule out the possibility of having minimum distance equal $3\delta$, any three vectors from the set of vectors
$$
  \{ (a_{i1}, a_{i2}, a_{i3}): i=1,2,\ldots, k+3 \}
$$
must be linearly independent. The $3\times 6$ matrix formed by taking these $k+3$ vector as the columns is a generator matrix of an MDS code of dimension 3 over $GF(4)$. Hence the number of repair groups are restricted to be less than or equal to~6. The parameters of optimal LRC in this case are
\begin{equation}
n=(k+3)\delta, \ k =2,3, \ r=1,\ \delta > 2,\ d=4\delta.
\label{LRC:r=1_d=4delta}
\end{equation}
An LRC in this category with 6 repair groups can be defined by a parity-check matrix as given in~\eqref{eq:H_r=1_a} with $\ell=6$ and
\begin{gather*}
A_1 = \begin{bmatrix}
0&0&\cdots &0 & 1 \\
0&0&\cdots &0 & 0 \\
0&0&\cdots &0 & 0 
\end{bmatrix},\ 
A_2 = \begin{bmatrix}
0&0&\cdots &0 & 1 \\
0&0&\cdots &0 & 1 \\
0&0&\cdots &0 & 1 
\end{bmatrix},\ 
A_3 = \begin{bmatrix}
0&0&\cdots &0 & 1 \\
0&0&\cdots &0 & \alpha\\ 
0&0&\cdots &0 & \alpha^2 
\end{bmatrix}, \\
A_4 = \begin{bmatrix}
0&0&\cdots &0 & 1 \\
0&0&\cdots &0 & \beta \\
0&0&\cdots &0 & \beta^2 
\end{bmatrix},\ 
A_5 = \begin{bmatrix}
0&0&\cdots &0 & 0 \\
0&0&\cdots &0 & 0 \\
0&0&\cdots &0 & 1 
\end{bmatrix},\ 
A_6 = \begin{bmatrix}
0&0&\cdots &0 & 0 \\
0&0&\cdots &0 & 1\\ 
0&0&\cdots &0 & 0 
\end{bmatrix}.
\end{gather*}

\section{Optimal quaternary $(r,\delta)$-LRC for $r>1$ and $d=3$}
\label{sec:d=3}

When the minimum distance is 3, we have
\begin{equation}
d = 3 = n - k  - \big( \big\lceil \frac{k}{r} \big\rceil -1 \big)(\delta-1) + 1.
\label{eq:d3}
\end{equation}
We let $\cD_1, \cD_2,\ldots, \cD_M$ be a chain consisting of $M=\lceil k/r \rceil$ subcodes in $\cC^\perp$. Since $\cC$ is optimal, the dimension of the direct sum $\cD_1\oplus\cdots\oplus\cD_M$ is equal to $M(\delta-1)$, which must be less than or equal to $n-k$,
$$
\big\lceil \frac{k}{r} \big\rceil (\delta-1) \leq n-k =\big( \big\lceil \frac{k}{r} \big\rceil -1 \big)(\delta-1) + 2.
$$
This yields $\delta \leq 3$. Because we consider LRC with $\delta>2$, the value of the parameter $\delta$ is restricted to 3 when $d=3$. Any subcode in $\cC^\perp$ satisfying Conditions (ii') and (iii') has dimension~2.

Since the longest quaternary MDS code with minimum distance $\delta=3$ has length 5, the size of a repair group is upper bounded by 5. The value of $r$ is equal to either 2 or 3.



We distinguish the following cases.

\smallskip
\underline{\em Case 1, $r=2 \text{ or } 3$ and $k = 0 \bmod r$}. Since $r|k$, all repair groups are mutually disjoint. Each repair group has size $r+\delta-1=r+2$. Let $\ell$ denote $n/(r+3-1)$. From \eqref{eq:d3}, we obtain $k=r\ell$.

When $r=2$, the code parameters of an optimal LRC are
\begin{equation}
n=4\ell, \ k =2\ell, \ r=2,\ \delta =3,\ d=3\\ \   (\ell \geq 2).
\label{LRC:r=2_delta=3_d=3}
\end{equation}
We can take the tensor product 
$$H=I_\ell \otimes 
\begin{bmatrix}
0&1&1&1 \\
1&0&1&\alpha
\end{bmatrix}
$$
as a parity-check matrix.

When $r=3$, the code parameters of an optimal LRC are
\begin{equation}
n=5\ell, \ k =3\ell, \ r=3,\ \delta =3,\ d=3\ \ \  (\ell \geq 2).
\label{LRC:r=3_delta=3_d=3}
\end{equation}
We can construct a code with these parameters by the parity-check matrix
$$
H=I_\ell \otimes 
\begin{bmatrix}
0&1&1&1&1 \\
1&0&1&\alpha&\alpha^2
\end{bmatrix}.
$$

\smallskip

\underline{\em Case 2, $r=3$ and  $k=2\bmod r$}. We write $k=3\ell-1$ for some integer~$\ell$. From \eqref{eq:d3} we deduce that the code length should be $n=5\ell-1$. The parameters of an optimal LRC in this case are
\begin{equation}
n=5\ell-1, \ k =3\ell-1, \ r=\delta =d=3.\ (\ell\geq 2)
\label{LRC:r=3_k=2_delta=3_d=3}
\end{equation}
An optimal code for $\ell=3$ can be constructed from a parity-check matrix in the form
$$
H= \left[
\begin{array}{cccccccccccccc}
0&1&1&1&1 \\
1&0&1&\alpha&\alpha^2 \\
&&&&0&1&1&1&1\\
&&&&1&0&1&\alpha&\alpha^2 \\
&&&&&&&&&0&1&1&1&1\\
&&&&&&&&&1&0&1&\alpha&\alpha^2
\end{array} \right]
$$
The first two repair groups overlap in one position, while each of the remaining $\ell-2$ repair groups is mutually disjoint with all other repair groups. One can check that the minimum distance is 3, since any pair of columns in $H$ are linearly independent.

\smallskip
\underline{\em Case 3, $r=2$ and  $k=1\bmod r$}. The derivation is similar to the previous case. Write $k=2\ell-1$ for some integer~$\ell$. The code parameters are
\begin{equation}
n=4\ell-1, \ k =2\ell-1, \ r=2,\ \delta =3,\ d=3.\ (\ell\geq 2)
\label{LRC:r=2_k=1_delta=3_d=3}
\end{equation}
An LRC with the above parameters can be constructed by an $2\ell\times(4\ell-1)$ parity-check matrix in which two repair groups overlap in exactly one position. An example for $\ell=3$ is shown below,
\begin{equation}
H= {\small \left[
\begin{array}{ccccccccccc}
0&1&1&1 \\
1&0&1&\alpha \\
&&&0&1&1&1\\
&&&1&0&1&\alpha \\
&&&&&&&0&1&1&1\\
&&&&&&&1&0&1&\alpha
\end{array} \right]}.
\label{eq:H_r=2_delta=3_d=3}
\end{equation}

\smallskip
\underline{\em Case 4, $r=3$ and  $k=1\bmod r$}. Write $k=3\ell-2$ for some integer~$\ell$. The code parameters are
\begin{equation}
n=5\ell-2, \ k =3\ell-2, \ r= \delta = d=3.\ (\ell\geq 2)
\label{LRC:r=3_k=1_delta=3_d=3}
\end{equation}
We can construct a parity-check matrix with two repair groups overlapping in exactly two positions. The following is an example for $\ell=3$,
\begin{equation}
H= {\small \left[
\begin{array}{ccccccccccccc}
0&1&1&1&1 \\
1&0&1&\alpha&\alpha^2 \\
&&&0&1&1&1&1\\
&&&1&0&1&\alpha&\alpha^2 \\
&&&&&&&&0&1&1&1&1\\
&&&&&&&&1&0&1&\alpha&\alpha^2
\end{array} \right]}.
\label{eq:H_r=3_delta=3_d=3}
\end{equation}

\section{Optimal Quaternary LRC with $r>1$ and $d=4$}
\label{sec:d=4}

Consider an optimal $(r,\delta)$-LRC $\cC$ with minimum distance is 4. We have 
\begin{equation}
d = 4 = n - k  - \big( \big\lceil \frac{k}{r} \big\rceil -1 \big)(\delta-1) + 1 .
\label{eq:d4}
\end{equation}
The residue code of an optimal $(r,\delta)$-LRC $\cC$ is an MDS code over $GF(4)$ with minimum distance 4. The parameters of the residue code are $[6,3,4]$, $[5,2,4]$, or $[4,1,4]$.

Because an optimal LRC restricted to a repair group is an MDS with minimum distance $\delta$, and the cases $r=1$ and $\delta=2$ are excluded, the value of $\delta$ is restricted to 3 or 4. 
 For $q=4$ and $\delta=4$, the largest size of a repair group is 6, the value of parameter $r$ in this case is less than or equal to 3. After excluding $r=1$ and $\delta=2$, the possible code parameters $(r,\delta)$ are $(2,3)$, $(3,3)$, $(2,4)$ and $(3,4)$.

\subsection{$r=2$, $\delta=3$}

\underline{\em Case 1: $k=0\bmod 2$}

Let $\ell=k/2$. By \eqref{eq:d4}, a potential optimal LRC in this case has length
$$
 n =3+2\ell+2(\ell-1) = 4\ell+1,
$$
which is not divisible by $r+\delta-1=4$. However, since $r|k$, by~\cite[Theorem 2.2]{Kamath14}, $n$ should be divisible by 4. This contradiction rules out the existence of optimal $(2,3)$-LRC with minimum distance 4 when $r|k$.

\underline{\em Case 2: $k=1\bmod 2$}

Write $k=2\ell-1$ for some integer $\ell$. By \eqref{eq:d4}, the length $n$ is
$$
n = 3+(2\ell-1)+2(\ell-1) = 4\ell.
$$
The code parameters of an optimal LRC in this case are
\begin{equation}
n=4\ell, \ k =2\ell-1, \ r=2,\ \delta =3,\ d=4\ \ (\ell\geq 2).
\label{LRC:r=2_k=1_delta=3_d=4}
\end{equation}
We can construct LRC with these parameters using the generalized tensor product construction. Let $B_1$ and $B_2$ be the matrices
$$
B_1 = \begin{bmatrix} 1&1&1&1 \\0&1&\alpha&\alpha^2\end{bmatrix}, \
B_2 = \begin{bmatrix} 0&1&\alpha^2&\alpha \end{bmatrix}.
$$
The matrix $B_1$ is a parity-check matrix of a $[4,2,3]$ MDS code, while the concatenated matrix $\begin{bmatrix}B_1\\B_2\end{bmatrix}$ is a parity-check matrix of a $[4,1,4]$. 
Let $\mathbf{1}_\ell$ be an $\ell\times 1$ all-one row vector. By Prop.~\ref{prop:GTP}, the code $\cC(\ell,B_1,\mathbf{1}_{\ell},B_2)$ is a $(2,3)$-LRC with minimum distance~4.

\subsection{$r=3$, $\delta=3$}

\underline{\em Case 1: $k=2\bmod r$}

Let $k=3\ell-1$ for some integer $s$. The code length $n$ is equal to
$$
n = 3+(3\ell-1)+2(\ell-1) = 5\ell
$$
by \eqref{eq:d4}. The code parameters of an optimal LRC in this case are
\begin{equation}
n=5\ell, \ k =3\ell-1, \ r=3,\ \delta =3,\ d=4\ \ (\ell\geq 2).
\label{LRC:r=3_k=1_delta=3_d=4}
\end{equation}

We can construct optimal LRCs with the above parameters using the matrix
\begin{equation} B=
\begin{bmatrix}
1&1&1&1&0 \\
0&1&\alpha&\alpha^2&1\\
0&1&\alpha^2&\alpha&0
\end{bmatrix},
\label{eq:MDS523}
\end{equation}
which is a parity-check matrix of a quaternary $[5,2,4]$ MDS code.
Let $B_1$ be the $2\times 5$ consisting of the first two rows of $B$. Using $B_1$ as a parity-check matrix, we can construct a quaternary $[5,3,3]$ MDS code. We can apply the generalized tensor product construction to obtain 
an optimal $(3,3)$-LRC $\cC(\ell,B_1,\mathbf{1}_\ell, B_2)$ with minimum distance 4.

\smallskip
\underline{\em Case 2: $k=1\bmod r$}

If $k=1\bmod 3$, we write $k$ as $3\ell-2$ and check that the code length should be
$$
n = 3+(3\ell-2)+2(\ell-1) = 5\ell-1.
$$
The code parameters are
\begin{equation}
n=5\ell-1, \ k =3\ell-2, \ r=3,\ \delta =3,\ d=4\ \ (\ell\geq 2).
\label{LRC:r=3_k=1_delta=3_d=4b}
\end{equation}
The construction is similar to the previous case, except that the last repair group on the lower right corner of the parity-check matrix is ``shifted to the left'', such that it overlaps with the second last repair group at one position. An example with $\ell=3$ repair groups is shown below,
$$ H= {\small \left[
\begin{array}{ccccc|ccccc|cccc}
1&1&1&1&0 \\
0&1&\alpha&\alpha^2&1 \\ \hline
&&&&&1&1&1&1&0 \\ 
&&&&&0&1&\alpha&\alpha^2&1 \\ \hline
&&&&&&&&&1&1&1&1&0 \\
&&&&&&&&&0&1&\alpha&\alpha^2&1 \\ \hline
0&1&\alpha^2&\alpha&0&0&1&\alpha^2&\alpha&0&1&\alpha^2&\alpha&0
\end{array} \right]}.
$$
It defines a $(3,3)$-LRC with length 15, dimension 7, and minimum distance 4.

\smallskip

\underline{\em Case 3: $k=0\bmod r$}

When $k$ is divisible by 3, we can express $k$ as $3\ell$. From \eqref{eq:d4}, we see that the length of an optimal LRC 
$$
n = 3+3\ell+2(\ell-1) = 5\ell+1
$$
is not divisible $5=(r+\delta-1)$.  An optimal LRC would violate ~\cite[Theorem 2.2]{Kamath14}, which says that the code length should be divisible by $(r+\delta-1)$ when $r|k$. Hence there is no Singleton-optimal LRC in this case.

\subsection{$r=2$, $\delta=4$}

Write $k=2\ell-e$, where $e=0,1$.  We note that for both $e=0,1$, we have $\lfloor (k-1)/r\rfloor = \ell-1$.
the code length is
$$
n = 3+(2\ell-e)+3(\ell-1) = 5\ell-e.
$$
The code parameters are
\begin{equation}
n=5\ell-e, \ k =2\ell-e, \ r=2,\ \delta =4,\ d=4\ \ (\ell\geq 2,\ e=0,1).
\label{LRC:r=2_delta=4_d=4}
\end{equation}

We can use the matrix $B$ in \eqref{eq:MDS523} to construct LRCs with the above code parameters. For even $k=2\ell$, we have $\ell$ mutually disjoint repair groups. We can take the tensor product $I_{\ell}\otimes B$ as a parity-check matrix. For odd $k=2\ell-1$,
two repair groups overlap at exactly one position, while the remaining repair groups are mutually disjoint with the others. An example for $\ell=3$ is shown below,
$$
H = {\small  \left[
\begin{array}{cccccccccccccc}
1&1&1&1&0 \\
0&1&\alpha&\alpha^2&1 \\
0&1&\alpha^2&\alpha&0 \\
&&&&1&1&1&1&0 \\
&&&&0&1&\alpha&\alpha^2&1 \\
&&&&0&1&\alpha^2&\alpha&0 \\
&&&&&&&&&1&1&1&1&0 \\
&&&&&&&&&0&1&\alpha&\alpha^2&1 \\
&&&&&&&&&0&1&\alpha^2&\alpha&0 
\end{array} \right]}.
$$

\subsection{$r=3$, $\delta=4$}

Write $k=3\ell-e$, where $e$ is an integer equals to 0, 1 or 2. The code length is
$$
n = 3 + (3\ell-e) + 3(\ell-1) = 6\ell-e.
$$
The code parameters are
\begin{equation}
n=6\ell-e, \ k =3\ell-e, \ r=3,\ \delta =4,\ d=4\ \ (\ell\geq 2,\ e=0,1,2).
\label{LRC:r=3_delta=4_d=4}
\end{equation}

Let $D$ denote the $3\times 6$ parity-check matrix 
\begin{equation} D=
\begin{bmatrix}
1&1&1&1&0&0 \\
0&1&\alpha&\alpha^2&1&0\\
0&1&\alpha^2&\alpha&0&1
\end{bmatrix},
\label{eq:MDS634}
\end{equation}
When $n=6\ell$, we can take the tensor product $I_\ell \otimes D$ as the parity-check matrix. When $n=6\ell-1$ or $6\ell-2$, the structure of the parity-check matrix is similar to the matrices in \eqref{eq:H_r=2_delta=3_d=3} and~\eqref{eq:H_r=3_delta=3_d=3}.

\section{Optimal Quaternary $(r,\delta)$-LRC with $r>1$, $k=r+1$, $d>4$}
\label{sec:d>4a}

Prop.~\ref{prop:residue_code} says that if we restrict an optimal LRC to a repair group, the resulting punctured code is an MDS code with minimum distance $\delta$. For $q=4$, the parameters of MDS codes are: $[\delta,1,\delta]$, $[\delta+1,\delta,2]$, $[4,2,3]$, $[5,3,3]$, $[5,2,4]$ and $[6,3,4]$. Neglecting the code parameters with $r=1$ and $\delta=2$, there are only four combinations for the code parameters $(r,\delta)$, and they are $(2,3)$, $(2,4)$, $(3,3)$, and $(3,4)$. 


We need the following structural properties when its minimum distance is larger than the alphabet size. We note that parts (iii) and the conclusion in part (iv) also hold for optimal LRC under the condition $r|k$~\cite[Theorem 2.2]{Kamath14}.

\begin{prop}
Let $\cC$ be an optimal $q$-ary $k$-dimensional $(r,\delta)$-LRC with  minimum distance $d_{\mathrm{min}}(\cC)$ strictly larger than~$q$. Then
\begin{enumerate}
\item[(i)] a residue code of $\cC$ has dimension 1, and is equivalent to a repetition code;
\item[(ii)] $r | (k-1)$;
\item[(iii)] each repair group $\mathcal{R}$ in $\cC$ has size $r+\delta-1$;
\item[(iv)] furthermore, if $k>2r$, then the code length $n$ is divisible by $r+\delta-1$, and the code symbols can be partitioned into disjoint repair groups of size $r+\delta-1$.
\end{enumerate}
\label{prop:disjoint_support}
\end{prop}

\begin{proof} 

(i) Under the hypothesis that $\cC$ has minimum distance larger than $q$, the residue code is an MDS code with minimum distance larger than $q$. Part (i) then follows from Theorem~\ref{thm:defect}.

(ii) The property $r|(k-1)$ is a direct consequence of Theorem~\ref{thm:d}.

(iii) Let $\mathscr{D}$ be the set of all subcodes in $\cC^\perp$ that satisfy Conditions (b') and (c').
Suppose $\mathcal{R}_0$ is a repair group with size strictly less than $r+\delta-1$, and $\cD_0$ be the subcode in $\cC^\perp$ associated with $\cR_0$. We form a chain  as follows. We initialize by letting $\cD_0$ be the first subcode in the chain, and iteratively append a subcode to it if the resulting list satisfies the definition of chain. We stop when the chain is maximal. By Prop.~\ref{prop:chain}, we have at least $\lceil k/r \rceil$ subcodes in the chain. By part (ii), we know that $\lceil k/r \rceil - 1 = (k-1)/r$ is a positive integer. Since one of the subcode in the chain has support strictly less than $r+\delta-1$, the support of $\cD_1\oplus\cdots\oplus\cD_{(k-1)/r}$ is upper bounded by
\begin{equation}
|\supp(\cD_1\oplus\cdots\oplus\cD_{(k-1)/r})|
\leq \frac{k-1}{r}(r+\delta-1) - 1
\label{eq:disjoint_support_proof}
\end{equation}
From Prop.~\ref{prop:residue_code} we know that the co-dimension of the residue code is $n-k-\frac{k-1}{r}(\delta-1)$. Hence, the dimension of the residue code is lower bounded by
\begin{align*}
& [n- (\frac{k-1}{r}(r+\delta-1) - 1)] - [n-k - \frac{k-1}{r}(\delta-1)] \geq 2.
\end{align*}
We can now apply Theorem~\ref{thm:defect} and get a contradiction that $d_{\mathrm{min}}(\cC)\leq q$.

(iv) 
Suppose $k>2r$. Under this assumption, any maximal chain has at least three subcodes. If there are two overlapping repair groups, we can set up a chain with these two overlapping subcodes as the first two subcodes. The support size of the first $(k-1)/r$ subcodes in this chain is upper bounded by \eqref{eq:disjoint_support_proof}. We can repeat the same argument in part (iii) and show that the residue code has dimension at least 2. By Theorem~\ref{thm:defect}, we get a contradiction $d_{\mathrm{min}}(\cC)\leq q$. This shows that any two subcodes in $\mathscr{D}$ has mutually disjoint supports. Since all repair groups has the same size, namely $r+\delta-1$, the code length must be divisible  by $r+\delta-1$.
\end{proof}

In this section, we consider the case $k=r+1$. LRCs with dimension $k>r+1$ will be classified in the next section. In this section and the next section, the computations are performed in Sage~\cite{sagemath}.

When $k=r+1$, the Singleton-type bound says that an optimal LRC in this case is equal to
\begin{equation}
d = n - k - \lfloor (k-1)/r \rfloor (\delta-1) + 1 = n - k - \delta + 2.
\label{eq:d_AMDS}
\end{equation}
In other words, an optimal LRC has defect $\delta-1$.

\subsection{$r\in\{2,3\}$, $\delta=3$, $k=r+1$}

When $\delta=3$, an optimal LRC is a code with defect 2. We can look up quaternary codes with defect 2 from some online table, such as MinT \cite{MinT}. For dimension 4, the largest  code length is 17, and for dimension 3, the largest code length is 16.

For dimension $k=4$, the longest code with defect 2 is an ovoid code. It can also be constructed by a quaternary BCH code with length 17 and designed distance 12. To construct this BCH code, we can take a 17-th root of unity in $GF(256)$, say $\gamma$, define a quaternary BCH $\cC_{17}$ code with zeros 
$$
\gamma^{-5},\ \gamma^{-4},\ 
\gamma^{-3},\ \gamma^{-2},\ 
\gamma^{-1},\ \gamma^{0},\ \gamma^{1},\ \gamma^{2},\  
\gamma^{3},\ \gamma^{4},\ \gamma^5.
$$
The resulting code is a $[17,4,12]$ code. A generator matrix in cyclic form is
$$
G_{17} = \left[ \begin{array}{ccccccccccccccccc}
1&1&\alpha&0&1&\alpha&\alpha^2&\alpha^2&\alpha&1&0&\alpha&1&1&0&0&0 \\
0&1&1&\alpha&0&1&\alpha&\alpha^2&\alpha^2&\alpha&1&0&\alpha&1&1&0&0 \\
0&0&1&1&\alpha&0&1&\alpha&\alpha^2&\alpha^2&\alpha&1&0&\alpha&1&1&0 \\
0&0&0&1&1&\alpha&0&1&\alpha&\alpha^2&\alpha^2&\alpha&1&0&\alpha&1&1 \\
\end{array} \right].
$$
The repair groups form a $3$-(17,5,1) balanced incomplete block design (BIBD). This LRC has the special property that any three code symbols are contained in a unique repair group. Other optimal $(3,3)$-LRC of dimension 4 can be obtained by puncturing the code $\cC_{17}$. The details are tabulated in Table~\ref{table:puncture33}. We remark that the locality of the ovoid code $\cC_{17}$ is also studied in~\cite{FLY}.  We refer the readers to \cite{BritzShiromoto} for more details about designs from subcode supports.

\begin{table}
\caption{Construction of Optimal $(3,3)$-LRC of Dimension 4 for $10\leq n \leq 17$.}
\label{table:puncture33}
\begin{center}
\begin{tabular}{|c|c|c|l|} \hline
Code length & Dim. & Min. Dist. &Construction \\ \hline \hline
17 & 4& 12& $C_{17}$\\
16 & 4& 11&Puncture $C_{17}$ at positions 17 \\
15 & 4& 10&Puncture $C_{17}$ at positions 16,17 \\
14 & 4& 9&Puncture $C_{17}$ at positions 15,16,17 \\
13 & 4& 8&Puncture $C_{17}$ at positions 14,15,16,17 \\
12 & 4& 7&Puncture $C_{17}$ at positions 13,14,15,16,17 \\
11 & 4& 6&Puncture $C_{17}$ at positions 12,13,14,15,16,17 \\
10 &  4& 5&Puncture $C_{17}$ at positions 5,6,9,11,12,13,17 \\ \hline
\end{tabular}
\end{center}
\end{table}

We summarize the parameters of optimal $(3,3)$-LRC with dimension 4 below.
\begin{equation}
d+5, \ k =4, \ r=3,\ \delta =3,\ 5\leq d\leq 12.
\label{LRC:r=3_delta=3_d>4}
\end{equation}

By shortening $\cC_{17}$ at the last position, we can obtain an optimal quaternary $(2,3)$ code $\cC_{16}$ with dimension~3. It has a generator matrix
$$
G_{16} = \left[ \begin{array}{cccccccccccccccc}
1&1&\alpha&0&1&\alpha&\alpha^2&\alpha^2&\alpha&1&0&\alpha&1&1&0&0 \\
0&1&1&\alpha&0&1&\alpha&\alpha^2&\alpha^2&\alpha&1&0&\alpha&1&1&0 \\
0&0&1&1&\alpha&0&1&\alpha&\alpha^2&\alpha^2&\alpha&1&0&\alpha&1&1 \\
\end{array} \right].
$$
One can check that there are 20 repair groups and they form a $2$-$(16,4,1)$ BIBD, i.e., any two code symbols belong to a unique repair group. The LRC $\cC_{16}$ has very large availability. Using the balanced property of BIBD, each code symbol, say at position $i$, is covered by precisely 5 repair groups, say $\cR_1$, $\cR_2$, $\cR_3$, $\cR_4$, and $\cR_5$. If we remove the index $i$ from these 5 repair groups, we will get 5 mutually disjoint sets, i.e.,
$$
(\cR_a \setminus \{i\}) \cap (\cR_b \setminus \{i\}) = \emptyset 
$$
for $a\neq b$.

By puncturing, we can obtain optimal 3-dimensional $(2,3)$-LRC with length $9\leq n \leq 15$. The puncturing patterns are shown in Table~\ref{table:puncture23}. 

\begin{table}
\caption{Construction of Optimal $(2,3)$-LRC of Dimension 3 for $9\leq n \leq 16$.}
\label{table:puncture23}
\begin{center}
\begin{tabular}{|c|c|c|l|} \hline
Code length & Dim. & Min. Dist. &Construction \\ \hline \hline
16 & 3& 12& $C_{16}$  \\
15 & 3& 11&Puncture $C_{16}$ at position 1 \\
14 & 3& 10&Puncture $C_{16}$ at positions 1,2 \\
13 & 3& 9&Puncture $C_{16}$ at positions 1,2,3 \\
12 & 3& 8&Puncture $C_{16}$ at positions 1,2,3,4 \\
11 & 3& 7&Puncture $C_{16}$ at positions 1,2,3,4,5 \\
10 & 3& 6&Puncture $C_{16}$ at positions 1,2,3,4,5,8 \\
9 &  3& 5&Puncture $C_{16}$ at positions 1,2,3,10,12,16 \\ \hline
\end{tabular}
\end{center}
\end{table}

The parameters of optimal $(2,3)$-LRC with dimension 3 are
\begin{equation}
n=d+4, \ k =3, \ r=3,\ \delta =3,\ 5\leq d\leq 12.
\label{LRC:r=2_delta=3_d>4}
\end{equation}

\subsection{$r\in\{2,3\}$, $\delta=4$, $k=r+1$}

Optimal $(r,4)$-LRC with $k=r+1$ has defect 3. The longest code length for dimension $k=3$ is 21, attained  by a simplex code over $GF(4)$.  The longest code length for dimension $k=4$ is 18~\cite{MinT}.

We can obtain a quaternary simplex code $\cC_{21}$ of length 21 by the following generator matrix,
$$
G_{21} = \left[ \begin{array}{cccccccccc cccccccccc c}
1&0&1& 1&1    &0   &1&0&1&1&1&1&0&1&1&1&1&0&1&1&1  \\
0&1&1&\alpha&\alpha^2& 0&0& 1& 1& \alpha&\alpha^2& 0& 1& 1& \alpha& \alpha^2& 0& 1& 1& \alpha& \alpha^2 \\
0& 0& 0& 0& 0& 1& 1& 1& 1& 1& 1& \alpha&\alpha&\alpha&\alpha&\alpha&\alpha^2&\alpha^2&\alpha^2&\alpha^2&\alpha^2
\end{array} \right].
$$
The columns are nonzero vectors in $GF(4)^3$. We normalize the generator matrix such that the top entry in each column is equal to~1.
This is an optimal quaternary $(2,3)$-LRC with dimension 3 and minimum distance~16.
 This code belong to a family of optimal $(2,q)$-LRC with length $q^2+q+1$ and minimum distance $q^2$~\cite{FLY}. The repair groups have constant size 5, and they form a $2$-$(21,5,1)$ BIBD. Each code symbol is covered by 5 repair groups. By puncturing at some appropriately chosen coordinates, we can obtain other optimal $(2,3)$-LRC of dimension 3. In Table~\ref{table:puncture24}, we illustrate how to obtain $(2,3)$-LRC with length 12 to 21. 

We next show that length $n=10$ and $n=11$ is not possible in this case. We need the following lemma.

\begin{lemma}
 In an optimal $(2,4)$-LRC with dimension 3, any two repair groups intersect at exactly one coordinate.
\label{lemma:repair_overlap1}
\end{lemma}

\begin{proof}
The proof relies on the weight distribution of the $[5,2,4]$ MDS code,
$$A_0=1,\ A_1=A_2=A_3=0,\ A_4=15,\ A_5=0.$$
There is no codeword with Hamming weight~5.

When $k=3$ an $r=2$, a maximal chain of subcodes contains at least $\lceil k/r \rceil=2$ subcodes.  Suppose there are two repair groups with disjoint supports. By permuting the coordinates we may assume that these two repair groups are located at positions 1 to 5, and positions 6 to~10. If we puncture the first repair group, the remaining part of the codeword is a repetition code. However, the weight distribution of the $[5,2]$ MDS code implies that one of the code symbol in the support of the second repair group must be zero, contradicting that code symbol in a repetition code are all nonzero.

Next suppose that there are two repair groups that overlap at exactly two positions. By permuting the coordinates we may assume that their supports are $\{1,2,3,4,5\}$ and $\{4,5,6,7,8\}$. Consider a codeword that is zero in the first 5  coordinates. Because there is no codeword of weight 3 in the $[5,2,4]$ MDS code, the coordinates at positions 5, 6, and 7 must be zero. However, the code symbols from position 5 to the end of the codeword should be codeword of a repetition code by part (i) of Prop.\ref{prop:disjoint_support}. Hence this codeword is an all-zero codeword.

Similarly we can prove that no two repair groups can intersect at 3 or 4 coordinates.
\end{proof}

As an example of the above lemma, the three repair groups for the code with length 12 in Table~\ref{table:puncture24} are $\{1,2,3,4,5\}$, $\{1,6,7,9,11\}$, and $\{2,6,8,10,12\}$. Any two of them intersect at exactly one coordinates.

Using Lemma~\ref{lemma:repair_overlap1}, we can verify that it is not possible to arrange three repair groups in 11 coordinates, such that each repair group has length 5 and any two of them intersect at 1 coordinates. Likewise, we can eliminate the possibility of length $n=10$ in this case.

Thus, the code parameters of optimal $(2,4)$-LRC with dimension 3 are 
\begin{equation}
n=d+5, \ k =3, \ r=2,\ \delta =4,\ 7\leq d\leq16.
\label{LRC:r=2_delta=4_d>4}
\end{equation}

\begin{table}
\caption{Construction of Optimal $(2,4)$-LRC of Dimension 3 for $12\leq n \leq 21$.}
\label{table:puncture24}
\begin{center}
\begin{tabular}{|c|c|c|l|} \hline
Code length & Dim. & Min. Dist. &Construction \\ \hline \hline
21 & 3& 16& $C_{21}$  \\
20 & 3& 15&Puncture $C_{21}$ at position 1 \\
19 & 3& 14&Puncture $C_{21}$ at positions 1,2 \\
18 & 3& 13&Puncture $C_{21}$ at positions 11,14,20 \\
17 & 3& 12&Puncture $C_{21}$ at positions 11,13,14,20 \\
16 & 3& 11&Puncture $C_{21}$ at positions 2,4,10,12,20 \\
15 & 3& 10&Puncture $C_{21}$ at positions 10,11,14,15,19,21 \\
14 &  3& 9&Puncture $C_{21}$ at positions 9,10,11,14,15,19,21 \\ 
13 &  3& 8&Puncture $C_{21}$ at positions 10,11,13,14,16,18,19,20 \\ 
12 &  3& 7&Puncture $C_{21}$ at positions 9,10,11,14,15,16,19,20,21 \\ 
\hline
\end{tabular}
\end{center}
\end{table}

We next turn to $(3,4)$-LRC with dimension 4. We first prove a property about the repair groups, which is analogous to Lemma~\ref{lemma:repair_overlap1}.

\begin{lemma}
In an optimal quaternary $(3,4)$-LRC with dimension 4, any two repair groups are either disjoint or overlap in exactly two position.
\label{lemma:overlap_2}
\end{lemma}

\begin{proof}
We utilize the weight distribution of the $[6,3,4]$ hexacode
$$
A_0 = 1,\ A_1=A_2=A_3=0, \ A_4=45,\ A_5=0,\ A_6=18. 
$$
A nonzero codeword of $[6,3,4]$ hexacode has weight equal to either 4 or~6.

Suppose we have two repair groups $\cR_1$ and $\cR_2$ in $\cC$ that intersect at exactly one position. By re-labeling the columns of the parity-check matrix, we assume without loss of generality that $\cR_1=\{1,2,3,4,5,6\}$ and $\cR_2=\{6,7,8,9,10,11\}$. 
Let $\mathbf{c}$ be a nonzero codeword that is zero in the first 6 positions. The remaining code symbols from position 7 to the end of the codeword form a codeword in the residue code. Hence, there is a codeword in the LRC that is nonzero at all positions from position 7 to the end. By the requirements of LRC, when we restrict to $\cR_2$, the punctured code  $\cC_{\cR_2}$ is a $[6,3,4]$ code. However, there is no codeword of weight 5 in $[6,3,4]$ code. This contradicts the weight distribution of the quaternary $[6,3,4]$ code.

Likewise, we can show that no two repair groups can overlap in 3, 4, or 5 positions.
\end{proof}

An example of parity-check matrix of an optimal $(3,4)$-LRC with dimension 4 and length 18 is given in Fig.~\ref{fig:18_4_12}. We denote this code by $\cC_{18}$. From online code table~\cite{MinT}, we check that the length of a quaternary code with defect 3 and dimension 4 is less than or equal to 18. Hence the code $\cC_{18}$ is indeed the longest one among all $(3,4)$-LRCs with dimension 4. 

Using computation mathematical software Sage~\cite{sagemath}, we can see that $\cC_{18}$   actually has one more repair group. The four repair groups of  $\cC_{18}$ are:
$$
\{1,2,3,4,5,6\},\ 
\{7,8,9,10,11,12\},\ 
\{13,14,15,16,17,18\},\ 
\{1,3,8,11,16,18\}.
$$
If we remove columns 13, 14, 15, and 17 in Fig.~\ref{fig:18_4_12}, the resulting matrix defines a $(3,4)$-LRC of length 14 and dimension 4. It has three repair groups
$$
\{1,2,3,4,5,6\},\ 
\{7,8,9,10,11,12\},\ 
\{1,3,8,11,16,18\}.
$$
If we remove the 6 columns 13 to 18 in Fig.~\ref{fig:18_4_12}, we will get a $(3,4)$-LRC of length 12 and dimension 4. It has two repair groups
$$
\{1,2,3,4,5,6\},\ 
\{7,8,9,10,11,12\}.
$$

We can obtain $(3,4)$-LRCs with length 16 and dimension 4 by the following parity-check matrices,
$$H_{16}= {\small \left[\begin{array}{cccccc|cccccc|cccc}
1&0&0&1&1&1& 0&0&0&0&0&0&  0&0&0&0 \\
0&1&0&1&\alpha&\beta& 0&0&0&0&0&0&  0&0&0&0 \\
0&0&1&1&\beta&\alpha& 0&0&0&0&0&0&  0&0&0&0 \\ \hline
0&0&0&0&0&0& 1&0&0&1&1&1&  0&0&0&0 \\
0&0&0&0&0&0& 0&1&0&1&\alpha&\beta&  0&0&0&0 \\
0&0&0&0&0&0& 0&0&1&1&\beta&\alpha&  0&0&0&0 \\ \hline
0&0&0&0&0&0& 0&0&0&0&1&0&  0&1&1&1 \\
0&0&0&0&0&0& 0&0&0&1&0&0&  0&1&\alpha&\beta \\
0&0&0&0&0&0& 0&0&0&0&0&0&  1&1&\beta&\alpha \\ \hline
0&0&0&1&\alpha&\alpha& 0&0&0&\alpha&\beta&0&  0&1&1&\beta \\
0&0&0&1&\beta&\beta& 0&0&0&\beta&0&\beta&  0&1&0&\alpha \\
0&0&0&\alpha&1&\beta& 0&0&0&\alpha&0&\alpha&  0&\alpha&1&\alpha
\end{array} \right]}.
$$
The last three rows correspond to global parity-check equations, while the first nine rows define the local codes. The repair groups are located at $\{1,2,3,4,5,6\}$, $\{7,8,9,10,11,12\}$, and $\{9,10,13,14,15,16\}$. This parity-check matrix defines an $(3,4)$-LRC of length 16, dimension 4 and minimum distance 10. 

In the remainder of this section we consider $(3,4)$-LRC with odd length. For length $n=17$, we can construct a $(3,4)$-LRC with dimension 4 and distance 11 by the following parity-check matrix,
$$
H_{17} = {\small  \left[\begin{array}{cccccc|cccccc|ccccc}
1&0&0&1&1&1& 0&0&0&0&0&0& 0&0&0&0&0 \\
0&1&0&1&\alpha&\beta& 0&0&0&0&0&0& 0&0&0&0&0\\
0&0&1&1&\beta&\alpha& 0&0&0&0&0&0& 0&0&0&0&0\\ \hline
1&0&0&0&0&0& 0&1&1&1&0&0& 0&0&0&0&0\\
0&1&0&0&0&0& 0&1&\alpha&\beta&0&0& 0&0&0&0&0\\
0&0&0&0&0&0& 1&1&\beta&\alpha&0&0& 0&0&0&0&0\\  \hline
1&0&0&0&0&0& 0&0&0&0&1&1& 1&0&0&0&0\\
0&0&1&0&0&0& 0&0&0&0&1&\alpha& \beta&0&0&0&0\\
0&0&0&0&0&0& 1&0&0&0&1&\beta& \alpha&0&0&0&0\\ \hline
0&0&0&0&0&0& 0&0&0&0&0&1& 0&0&1&1&1\\
0&0&0&0&0&0& 0&0&0&0&0&0& 1&0&1&\alpha&\beta\\
0&0&0&0&0&0& 0&0&0&0&0&0& 0&1&1&\beta&\alpha\\ \hline
\beta&\alpha&\alpha&\beta&1&\beta& 1&\alpha&\beta&\beta&\alpha&\alpha& \beta&\alpha&\beta&1&\beta
\end{array}\right]}.
$$
The elements of $GF(4)$ are represented by $\{0,1,\alpha,\beta\}$, where $\alpha$ and $\beta$ are the roots of $x^2+x+1$ in $GF(4)$.
The repair groups of this LRC are $\{1,2,3,4,5,6\}$, $\{1,2,7,8,9,10\}$ and $\{1,3,7,11,12,13\}$ and $\{12,13,14,15,16,17\}$. We note that two repair groups are either disjoint or overlapping on exactly two locations. The last row in the parity-check matrix $H_{17}$ is a global parity-check equation.

By Lemma~\ref{lemma:overlap_2}, for length $n=15$, there is essentially one way to place the repair groups such that each pair of them are either disjoint or overlapping in exactly two positions. After some code symbols permutations, we can write the repair groups as
$$
\{1,2,3,4,8,9\},
\{1,2,5,6,10,11\},
\{1,3,5,7,12,13\},
\{1,4,6,7,14,15\}.
$$
The four repair groups all intersect at the first code symbol. Beside the first position, each pair of repair groups intersect at one more position. Each repair group is associated with three parity-check equations. If we stack all the parity-check equations to form a $12\times 15$ matrix, the resulting matrix has the following structure,
$$ H_{15} = 
{\small 
\left[
\begin{array}{ccccccc|cccccccc}
1&1&0&0&0&0&0& 1&1&0&0&0&0&0&0\\
1&0&1&0&0&0&0& \alpha&\beta&0&0&0&0&0&0\\
1&0&0&1&0&0&0& \beta&\alpha&0&0&0&0&0&0\\ \hline
1&1&0&0&0&0&0& 0&0&1&1&0&0&0&0\\
1&0&0&0&1&0&0& 0&0&\alpha&\beta&0&0&0&0\\
1&0&0&0&0&1&0& 0&0&\beta&\alpha&0&0&0&0\\ \hline
1&0&\alpha&0&0&0&0& 0&0&0&0&1&1&0&0\\
1&0&0&0&\alpha&0&0& 0&0&0&0&\alpha&\beta&0&0\\
1&0&0&0&0&0&1& 0&0&0&0&\beta&\alpha&0&0\\ \hline
1&0&0&\alpha&0&0&0& 0&0&0&0&0&0&1&1\\
1&0&0&0&0&\alpha&0& 0&0&0&0&0&0&\alpha&\beta\\
1&0&0&0&0&0&1& 0&0&0&0&0&0&\beta&\alpha
\end{array} \right] }
$$
The rows of this matrix are linearly dependent. We can see this  by checking that the product
$$
(1,1,1,1,1,1,\beta,\beta,\beta,\beta,\beta,\beta) \cdot H_{15} 
$$
is equal to the zero vector. Hence, we can regard this matrix as a parity-check matrix and define a quaternary code with dimension~4. By the structure of this matrix, each code symbol is covered by a $[6,3]$ local code. We can check that the minimum distance is 9 by computer software~\cite{sagemath}.

For length $n=13$, we can remove the last four columns and last four rows in matrix $H_{17}$. The resulting matrix is a parity-check matrix of an $(3,4)$-LRC with length 13 and minimum distance~7.

We thus obtain the following code parameters for optimal $(3,4)$-LRC of dimension 4,
\begin{equation}
n=d+6, \ k =4, \ r=3,\ \delta =4,\ 6 \leq d \leq 12.
\label{LRC:r=3_delta=4_d>4}
\end{equation}

For length 11, it is not possible to accommodate repair groups that cover all 11 code symbols without violating the structure constraint in Lemma~\ref{lemma:overlap_2}. Hence there is no optimal $(3,4)$-LRC $\cC$ of dimension 4 and length~$11$.


\section{Optimal Quaternary $(r,\delta)$-LRC with $r>1$, $k\geq r+2$, $d>4$}
\label{sec:d>4b}

In this section we classify optimal quaternary $(r,\delta)$-LRC with minimum distance larger than 4 and $k \geq 2$. By Prop.~\ref{prop:disjoint_support}, all repair groups have constant size $r+\delta-1$, and furthermore the supports of the repair groups are mutually disjoint. 

Since an optimal LRC has $r|(k-1)$ when $d>4$, we write $k = r\sigma+1$ for some integer $\sigma$, so that $\lfloor (k-1)/r \rfloor =\sigma$. Let $\ell = n/(r+\delta-1)$ denote the number of repair groups. The case $\sigma=1$ is treated in the previous section. In this section we assume $\sigma\geq 2$. The minimum distance can be expressed as
\begin{equation}
 d= n-(r\sigma+1) - \sigma(\delta -1)+1 = (r+\delta-1)(\ell-\sigma).
\label{eq:d_k>2}
\end{equation}
The minimum distance must be a multiple of the repair group size.

\subsection{$r=2$, $\delta=3$}

By Prop.~\ref{thm:d} in Section~\ref{sec:preliminaries}, the minimum distance is upper bounded by $\delta q=12$. Hence
$$
 4(\ell-\sigma) = d\leq 12.
$$
The value of $\ell-\sigma$ is at most 3. Since we assume $d>4$ in this section, we only need to consider $\ell-\sigma=2$ or 3.

Let $H_{2\times 4}$ denote the matrix
\begin{equation}
H_{2\times 4} = \begin{bmatrix}
1 &0 & 1 & 1\\
0 & 1&\alpha^2&\alpha
\end{bmatrix}.
\label{eq:H24}
\end{equation}

A parity-check matrix of an optimal $(2,3)$-LRC in this category has the form
\begin{equation} H = \left[
\begin{array}{c}
I_\ell\otimes H_{2\times 4} \\
B
\end{array} \right],
\label{eq:H_r2_delta3}
\end{equation}
where $I_\ell\otimes H_{2\times 4}$ is a block diagonal matrix with $H_{2,4}$ in the diagonal.

\medskip

\underline{$\ell=\sigma+2$, $d=8$}

When $\ell-\sigma=2$, the submatrix $B$ at the bottom has 3 rows, because
$$
 n - k - \ell(\delta-1) = 4\ell - (2\sigma+1) - 2\ell = 2(\ell-\sigma)-1 = 3.
$$
After some row reductions, column permutations, and multiplication of columns by nonzero scalar, we can assume that $H$ has the following structure,
$$ H = {\small \left[
\begin{array}{cccc|cccc|cccc}
1&0&1&1&&&&& \\
0&1&\beta&\alpha& &&&&&\\ \hline
&&&&1&0&1&1 \\
&&&&0&1&\beta&\alpha& \\ \hline
&&&&&&&&1&0&1&1 \\
&&&&&&&&0&1&\beta&\alpha \\ \hline 
0&0&c_1&d_1& 0&0&e_1&f_1& 0&0&g_1&h_1 \\
0&0&c_2&d_2& 0&0&e_2&f_2& 0&0&g_2&h_2 \\
0&0&c_3&d_3& 0&0&e_2&f_2& 0&0&g_3&h_3
\end{array} \right] }
$$
where $\alpha$ is a root of $x^2+x+1$ in $GF(4)$ and $\beta=\alpha^2$, and $c_i$, $d_i$, $e_i$, $f_i$, $g_i$, $h_i$ are constants, for $i=1,2,3$. Because the restriction of a codeword to each repair group is a codeword in a $[4,2,3]$ MDS code, a nonzero codeword  cannot have either 1 or 2 nonzero code symbols in a repair group. In order to construct an LRC with minimum distance 8, we want to make sure that whenever a nonzero codeword has support lying within the union of two repair groups, then the Hamming weight of this codeword must be equal to~8, by excluding the possibility of weight 6 and 7.

For $j=1,2,\ldots, \ell$, let the two column vectors of dimension 3 at the bottom of the parity-check matrix under the $j$-th repair group be denoted by $\mathbf{u}_j$ and $\mathbf{v}_j$. In the above example, $\mathbf{u}_1$ is the column matrix $(c_1,c_2,c_3)^T$ and $\mathbf{v}_1$ is the column matrix $(d_1,d_2,d_3)^T$. For all $j=1,2,\ldots, \ell$, the vector $\mathbf{u}_j$ and $\mathbf{v}_j$ should both be nonzero, otherwise we will have a codeword of weight~3. Furthermore, they must be linearly independent, otherwise we will get a codeword of weight 4 whose support is a repair group. In the followings, we let $V_j$ denote the vector space spanned by $\mathbf{u}_j$ and $\mathbf{v}_j$, for $j=1,2,\ldots, \ell$.

To eliminate the possibility of weight-6 or weight-7 codeword, we consider the restriction of a codeword to a repair group. Such a restricted codeword should be a codeword of the $[4,2]$ MDS code. Let $\mathcal{L}$ denote the MDS code defined as the null space of $H_{2\times 4}$. Denote the resulting portion of the codeword by $(x,y,\phi,\theta)$. Let $U_3$ be the combinations of the third and fourth coordinates when $(x,y,\phi,\theta)$ is a nonzero codeword of weight 3 in $\mathcal{L}$.  We can express $\mathcal{U}_3$ as
$$
\mathcal{U}_3 := \{ (\phi,\theta) \in GF(4)^2:\, \exists x,y\in GF(4),\  (x,y,\phi,\theta)\in \mathcal{L} \text{ and } wt_H((x,y,\phi,\theta)=3\}.
$$
The weight distribution of the $[4,2]$ MDS code  $\mathcal{L}$ is
$$
A_0 = 1,\ A_1=A_2=0, \ A_3=12,\ A_4=3.
$$
 Hence $|\mathcal{U}_3|=12$.
For $j=1,\ldots, \ell$, let
$$
\mathcal{A}_j :=
\{\phi \mathbf{u}_j + \theta \mathbf{v}_j \in GF(4)^3:\, (\phi,\theta) \in \mathcal{U}_3 \}.
$$
If a vector in $\mathcal{A}_j$ belongs to the span $\langle \mathbf{u}_m, \mathbf{v}_m\rangle$ for some $m\neq j$, then we have a codeword of weight 6 or 7. Conversely, any codeword of weight 6 or 7 arises in this way. This gives the following design criterion: choose the vectors $\mathbf{u}_j$'s and $\mathbf{v}_j$'s so that they satisfy
\begin{enumerate}
\item[(i)] for $j=1,2,\ldots,\ell$, the two vectors $\mathbf{u}_j$ and $\mathbf{v}_j$ are  linearly independent (in particular, they are both nonzero);
\item[(ii)] for any $j\neq m \in\{1,2,\ldots,\ell\}$, the set $\mathcal{A}_j$ and the vector space $\langle \mathbf{u}_m, \mathbf{v}_m\rangle$ are mutually disjoint.
\end{enumerate}

Using the matrix $H_{2\times 4}$ in~\eqref{eq:H24} as the the parity-check matrix of the MDS code $\mathcal{L}$, the three codewords of weight 4 in $\mathcal{L}$ are
$$
(\beta,\beta,\alpha,1),\ (1,1,\beta,\alpha),\ (\alpha,\alpha,1,\beta).
$$
Furthermore,  under the assumption of Condition (i), $\mathcal{A}_j$ does not contain the zero vector in $GF(4)^3$. Hence, for each $j$, 
$$
V_j \setminus \mathcal{A}_j =
\langle \alpha \mathbf{u}_j + \mathbf{v}_j \rangle,
$$
which is a subspace in $GF(3)$ with dimension 1 containing the vector $\alpha \mathbf{u}_j+\mathbf{v}_j$. Intuitively speaking, for repair group $j$, only the nonzero vectors in $\langle \alpha \mathbf{u}_j + \mathbf{v}_j \rangle$ can be shared with other repair groups; the vectors in $\mathcal{A}_j$ are forbidden to appears in other repair groups.

The longest code with this structure can be obtained only when the vector subspace $\langle \alpha \mathbf{u}_j + \mathbf{v}_j \rangle$, for $j=1,2,\ldots, \ell$, are the same subspace. This is similar to the ``sunflower'' construction in~\cite{Gorla16, CFXHF21}. Since the sets $\mathcal{A}_j$ for $j=1,2,\ldots, \ell$ must be mutually disjoint, the number of repair groups is upper bounded by
$$
\frac{4^3-4 }{ 4^2-4}
= 60/ 12 = 5.
$$

An optimal $(2,3)$-LRC with dimension 7, distance 8 and 5 disjoint repair groups can indeed be realized. It is defined by the following parity-check matrix 
$$ {\small 
\left[
\begin{array}{cccc|cccc|cccc|cccc|cccc}
1&0&1&1& 0&0&0&0& 0&0&0&0& 0&0&0&0& 0&0&0&0 \\
0&1&\beta&\alpha& 0&0&0&0& 0&0&0&0& 0&0&0&0& 0&0&0&0 \\ \hline
0&0&0&0& 1&0&1&1& 0&0&0&0& 0&0&0&0& 0&0&0&0 \\
0&0&0&0& 0&1&\beta&\alpha& 0&0&0&0& 0&0&0&0& 0&0&0&0 \\ \hline
0&0&0&0& 0&0&0&0& 1&0&1&1& 0&0&0&0& 0&0&0&0 \\
0&0&0&0& 0&0&0&0& 0&1&\beta&\alpha& 0&0&0&0& 0&0&0&0 \\ \hline
0&0&0&0& 0&0&0&0& 0&0&0&0& 1&0&1&1& 0&0&0&0 \\
0&0&0&0& 0&0&0&0& 0&0&0&0& 0&1&\beta&\alpha& 0&0&0&0 \\ \hline
0&0&0&0& 0&0&0&0& 0&0&0&0& 0&0&0&0& 1&0&1&1 \\
0&0&0&0& 0&0&0&0& 0&0&0&0& 0&0&0&0& 0&1&\beta&\alpha \\ \hline
0&0&1&1& 0&0&1&1& 0&0&1&1& 0&0&0&1& 0&0&\alpha&0 \\
0&0&\beta&\alpha& 0&0&\beta&\alpha& 0&0&1&1& 0&0&1&\beta& 0&0&0&\beta \\
0&0&0&\alpha& 0&0&\beta&\beta& 0&0&\alpha&1& 0&0&1&1& 0&0&0&\alpha 
\end{array} \right]}.
$$

By puncturing the longest code in this category, we can obtain shorter codes. Their parameters are
\begin{equation}
n=4\ell, \ k =2\ell-3, \ r=2,\ \delta =3,\ d=8\ (\ell=3,4,5).
\label{LRC:r=2_delta=3_large1}
\end{equation}

\underline{$\ell=\sigma+3$, $d=12$}

When $\ell=\sigma+3$, a parity-check matrix can be written as in \eqref{eq:H_r2_delta3}, where $B$ is a submatrix with 5 rows. We continue with the notation in the previous subsection, but the column vectors $\mathbf{u}_j$ and $\mathbf{v}_j$ are now 5-dimensional. As in the previous case, the vectors $\mathbf{u}_j$ and $\mathbf{v}_j$ are linearly independent. We let $V_j$ denote the vector space spanned by $\mathbf{u}_{j}$ and $\mathbf{v}_{j}$, for $j=1,2,\ldots, \ell$. In order to prevent a codeword with weight 8, the intersection of $V_j$ and $V_m$ should be $\{\mathbf{0}\}$ whenever $j\neq m$.

We continue with the notation as in the previous case, except that the set $\mathcal{A}_j$, for $j=1,2,\ldots, \ell$, is re-defined as a subset of vectors in $GF(4)^5$,
$$
\mathcal{A}_j := \{ \phi \mathbf{u}_j + \theta \mathbf{v}_j\in GF(4)^5:\, (\phi,\theta)\in\mathcal{U}_3\}.
$$

We list the requirement on $\mathbf{u}_j$ and $\mathbf{v}_j$ below:
\begin{enumerate}
\item[(i)] for $j=1,2,\ldots,\ell$, the vector space $V_j$ has dimension 2;
\item[(ii)] for any $j\neq m \in\{1,2,\ldots,\ell\}$, the intersection of $V_j$ and $V_m$ is $\{\mathbf{0}\}$;
\item[(iii)] for any three distinct $j_1,j_2,j_3\in\{1,2,\ldots, \ell\}$, the set $\mathcal{A}_{j_1}$ of vectors is not contained in the sum space of $V_{j_2}$ and $V_{j_3}$.
\end{enumerate}

The second condition ensures that any codeword with support confined within two repair groups must be the zero codeword. The third condition guarantees that if
a codeword with support confined in three repair groups, then all 12 code symbols in the three repair groups must be all nonzero.

With the same argument as in the previous case, the set $\mathcal{A}_j$ does not contain the zero vector. The set difference $V_j\setminus\mathcal{A}_j$ is a 1-dimensional subspace in $GF(4)^5$, spanned by the vector $\alpha\mathbf{u}_j+\mathbf{v}_j$. We let $\mathbf{w}_j$ denote the vector $\alpha\mathbf{u}_j+\mathbf{v}_j$. The vectors $\mathbf{w}_j$ and $\mathbf{u}_j$ form a basis of $V_j$, with $\mathbf{w}_j\in V_j\setminus\mathcal{A}_j$ and $\mathbf{u}_j\in\mathcal{A}_j$. 



One method for constructing optimal LRCs in this category is to put all vectors $\mathbf{w}_j$'s in a subspace of dimension 2, so that any three of them are linearly dependent, corresponding to a codeword of weight~12. Meanwhile, the sets of vectors $\mathcal{A}_j$, for $j=1,2,\ldots, \ell$, should all lie outside this 2-dimensional plane, so that no vector in $\mathcal{A}_j$ lies on this special 2-dimensional subspace. As each repair group occupies a distinct line in this subspace of dimension 2, largest number of repair groups obtained by this method is upper bounded by
$$
 \ell \leq \frac{|GF(4)|^2-1}{|GF(4)|-1} = 15/3 =  5.
$$

Such a $(2,3)$-LRC with 5 repair groups and minimum distance 12 exists and can be constructed from the following parity-check matrix,
$$ H= {\small \left[
\begin{array}{cccc|cccc|cccc|cccc|cccc}
1&0&1&\alpha&&&&&&&&&&&&& \\
0&1&1&\beta &&&&&&&&&&&& \\ \hline
&&&&1&0&1&\alpha&&&&&&&&& \\
&&&&0&1&1&\beta &&&&&&&&&\\ \hline
&&&&&&&&1&0&1&\alpha&&&&& \\
&&&&&&&&0&1&1&\beta &&&&& \\ \hline
&&&&&&&&&&&&1&0&1&\alpha& \\
&&&&&&&&&&&&0&1&1&\beta & \\ \hline
&&&&&&&&&&&&&&&&1&0&1&\alpha \\
&&&&&&&&&&&&&&&&0&1&1&\beta \\ \hline
0&0&0&1& 0&0&1&1& 0&0&\beta&\alpha& 0&0&\alpha&\beta& 0&0&\beta&\alpha \\
0&0&1&1& 0&0&1&0& 0&0&\alpha&\beta& 0&0&\beta&1& 0&0&\alpha&1 \\
0&0&\beta&\beta& 0&0&\beta&\beta& 0&0&\alpha&\alpha& 0&0&\alpha&\alpha& 0&0&1&1 \\
0&0&\alpha&\alpha& 0&0&\alpha&\alpha& 0&0&\beta&\beta& 0&0&\alpha&\alpha& 0&0&1&1 \\
0&0&\beta&\beta& 0&0&\alpha&\alpha& 0&0&\beta&\beta& 0&0&\beta&\beta& 0&0&\beta&\beta
\end{array} \right]}.
$$

An optimal LRC with four repair groups can be obtained by puncturing one repair group from the above LRC. The parameters of $(2,3)$-LRC with minimum distance 12 that we can construct are
\begin{equation}
n=4\ell, \ k =2\ell-5, \ r=2,\ \delta =3,\ d=12\ (\ell=4,5).
\label{LRC:r=2_delta=3_large2}
\end{equation}

\subsection{$r=3$, $\delta=3$}

As in the case $(r,\delta)=(2,3)$, the minimum distance when $r=3$ and $\delta=3$ is upper bounded by $12$. Together with \eqref{eq:d_k>2},
$$
d = (r+\delta-1)(\ell-\sigma) = 5(\ell-\sigma)
$$
The difference $\ell-\sigma$ is equal to either 1 or 2.

\underline{$\ell=\sigma+1$, $d=5$} 

In terms of the variable $\ell=\sigma+1$, the dimension $k$ of an optimal LRC in this case is given by
$$
k = 3\sigma+1 = 3(\ell-1)+1 = 3\ell-2.
$$
The possible code parameters are
\begin{equation}
n=5\ell, \ k =3\ell-2, \ r=3,\ \delta =3,\ d=5\ (\ell\geq 3).
\label{LRC:r=3_delta=3_large1}
\end{equation}

Indeed, all of the above code parameters can be realized by the generalized tensor product construction. Let $H_1$ be the top submatrix of the $3\times 5$ matrix
$$ M= \left[
\begin{array}{ccccc}
1&1&1&1 &0\\
0&1&\alpha&\alpha^2&1 \\ \hline
0&1&\alpha^2&\alpha&0 \\
0&1&1&1&1 
\end{array} \right]
$$
and $H_2$ be the submatrix at the bottom. The matrix $H_1$ is a parity-check matrix of a $[5,2,3]$ MDS code and the matrix $M$ above is a parity-check matrix of a $[5,1,4]$ MDS code. Using the generalized tensor product construction, the code $\cC(\ell,H_1,\mathbf{1}_\ell, H_2)$
is a $(3,3)$-LRC with length $5\ell$, dimension $3\ell-2$ and minimum distance 5, for all $\ell \geq 2$.

\medskip

\underline{$\ell=\sigma+2$, $d=10$}

The existence of optimal LRC in this category is open, but we know the structure of any optimal code if it exists.

 Since we consider $k\geq r+2$ in this section, the smallest dimension in this case is $k=2r+1= 7$, and there are at least $\ell=4$ repair groups. We will prove that that there is no $(3,3)$-LRC with $d=10$ and 4 repair groups. It then follows that $(3,3)$-LRC with $d=10$ and larger than 4 repair groups does not exist.

Let $H_{2\times 5}$ be the matrix
\begin{equation}
H_{2\times 5} := 
\begin{bmatrix}
1&0&1&1&1 \\
0&1&1&\alpha&\alpha^2
\end{bmatrix}.
\label{eq:H_25}
\end{equation}
The corresponding generator matrix is
\begin{equation}
G_{3\times 5} := 
\begin{bmatrix}
1&1&1&0&0 \\
1&\alpha&0&1&0\\
1&\alpha^2&0&0&1
\end{bmatrix}.
\label{eq:G_35}
\end{equation}

Suppose $\cC$ is an optimal $(3,3)$-LRC with minimum distance 10, defined by the following parity-check matrix,
\begin{equation}
H= \left[
\begin{array}{cccc}
H_{2\times 5} & & \\
&H_{2\times 5}  & \\
&&H_{2\times 5}  \\
&&&H_{2\times 5} \\ \hline
M_1 & M_2 & M_3 & M_4
\end{array} \right].
\label{eq:H_10}
\end{equation}
The matrix $M_j$, for $j=1,2,\ldots, 4$, is a $5\times 5$ matrix in which the first and second columns are zero,
$$
M_j = \begin{bmatrix}
0&0&u_1&v_1&w_1\\
0&0&u_2&v_2&w_2\\
0&0&u_3&v_3&w_3\\
0&0&u_4&v_4&w_4\\
0&0&u_5&v_5&w_5
\end{bmatrix}.
$$

The null space of $H_{2\times 5}$ is a $[5,3,3]$ MDS code, whose weight distribution is
$$
A_0=1,\ A_1=A_2=0,\ A_3 =30, \ A_4= 15,\ A_5=18.
$$

Let $\mathbf{u}_j$, $\mathbf{v}_j$ and $\mathbf{w}_j$ be the third, fourth and the fifth columns, respectively, in the matrix $M_j$, for $j=1,2,3,4$. Using the argument as in the case $(r,\delta)=(2,3)$, the three vectors $\mathbf{u}_j$, $\mathbf{v}_j$ and $\mathbf{w}_j$ must be linearly independent. 
Let $V_j$ denote the vector subspace spanned by  $\mathbf{u}_j$, $\mathbf{v}_j$ and $\mathbf{w}_j$, for $j=1,2,3,4$. 

Because $V_1$ and $V_2$ are 3-dimensional subspaces in  a 5-dimensional vector space, the dimension of their intersection $V_1\cap V_2$ cannot be the zero vector space. The dimension of $V_1\cap V_2$ should equal 1, 2 or 3. If $\dim(V_1\cap V_2)$ equals 2 or 3, we have two linearly independent codewords of weight 10 whose supports fall within  the first two repair groups. An appropriate linear combination of them is a codeword of weight less than 10, contradicting the assumption that $\mathcal{C}$ has minimum distance is 10. 

As a result, we mus have $\dim(V_i\cap V_j)=1$ for all distinct indices $i$ and $j$. It means that given any two repair groups, we can find a codeword of weight 10 supported in these two repair groups. 

However, we are not able to find any optimal $(3,3)$-LRC with dimension 7 and minimum distance 10. The best code that we can construct has minimum distance 9. 
The following parity-check matrix is an example of $(3,3)$ LRC with length  $n=20$, dimension $k=7$, and minimum distance $d = 9$.
$$
H= {\small  \left[
\begin{array}{ccccc|ccccc|ccccc|ccccc}
1&0&1&1&1& 0&0&0&0&0& 0&0&0&0&0& 0&0&0&0&0 \\
0&1&1&\alpha&\beta& 0&0&0&0&0& 0&0&0&0&0& 0&0&0&0&0 \\ \hline
        0&0&0&0&0& 1&0&1&1&1& 0&0&0&0&0& 0&0&0&0&0 \\
        0&0&0&0&0& 0&1&1&\alpha&\beta& 0&0&0&0&0& 0&0&0&0&0 \\ \hline
        0&0&0&0&0& 0&0&0&0&0& 1&0&1&1&1& 0&0&0&0&0 \\
        0&0&0&0&0& 0&0&0&0&0& 0&1&1&\alpha&\beta& 0&0&0&0&0 \\ \hline
        0&0&0&0&0& 0&0&0&0&0& 0&0&0&0&0& 1&0&1&1&1 \\
        0&0&0&0&0& 0&0&0&0&0& 0&0&0&0&0& 0&1&1&\alpha&\beta \\ \hline
        0&0&1&1&1& 0&0&1&1&1& 0&0&1&1&1& 0&0&1&1&1 \\
        0&0&\beta&\beta&\alpha& 0&0&\beta&\beta&\alpha& 0&0&1&1&\beta& 0&0&1&\beta&1 \\
        0&0&1&0&\alpha& 0&0&\alpha&1&\alpha& 0&0&\alpha&\alpha&1& 0&0&\alpha&\alpha&1 \\
        0&0&1&\beta&1& 0&0&1&0&\alpha& 0&0&0&1&1& 0&0&\alpha&1&\alpha \\
        0&0&\alpha&\beta&0& 0&0&\beta&\beta&\alpha& 0&0&\alpha&\beta&0& 0&0&1&1&\beta 
\end{array} \right]}
$$
where $\alpha$ and $\beta$ are the roots of $x^2+x+1$ in $GF(4)$. There are 108 codewords of weight 9 in this LRC. The support of each of them is spread over three repair groups ($3+3+3$).

We conjecture that it is not possible to have minimum distance 10 in this case, and the largest minimum distance that we can get is $d=9$, but we do not have a proof yet.

\subsection{$r=2$, $\delta=4$}

From Prop~\ref{prop:disjoint_support}, the code symbols are covered by disjoint repair groups, and each repair group has size 5. 

Recall that there is no codeword of weight 5 in a $[5,2,4]$ MDS code. This means that if we restrict the LRC to a repair group, the local code is either the all zero codeword of length 5, or a codeword of Hamming weight 4. Regardless of the number of repair groups, a nonzero codeword in the residue code should not have any zero component. This is not compatible with the weight distribution of the $[5,2,4]$ MDS code. Therefore, there is no optimal LRC with $r=2$, $\delta=4$ and $q>4$.

\subsection{ $r=3$, $\delta=4$} 

Since the minimum distance is upper bounded by 16, from \eqref{eq:d_k>2}, we get
$$
d = (r+\delta-1)(\ell-\sigma) = 6(\ell-\sigma) \leq 16.
$$
The value of $\ell-\sigma$ is equal to either 1 or 2. 

\medskip

\underline{$\ell=\sigma+1$, $d=6$}

The dimension $k$ of an optimal LRC is
$$
k = 3\sigma+1 = 3(\ell-1)+1 = 3\ell-2.
$$
The possible code parameters are
\begin{equation}
n=6\ell, \ k =3\ell-2, \ r=3,\ \delta =4,\ d=6\ (\ell\geq 3).
\label{LRC:r=3_delta=4_large1}
\end{equation}

All of the above code parameters can be realized by the generalized tensor product construction. Consider the $3\times 6$ matrix
$$ M= \left[
\begin{array}{cccccc}
1&0&0 &1&1&1 \\
0&1&0 &1&\alpha&\alpha^2 \\ 
0&0&1 &1&\alpha^2&\alpha \\ \hline
1&\alpha&\alpha^2&1&\alpha&\alpha^2 \\
1&\alpha^2&\alpha^2&\alpha&\alpha&1
\end{array} \right]
$$
Let $H_1$ be the submatrix consisting of the first three rows in $M$, and 
$H_2$ be the submatrix consisting of the last two rows.
The matrix $H_1$ is a parity-check matrix of a $[6,3,4]$ code. The second submatrix $H_2$ is chosen such that $M$ is the parity-check matrix of a $[6,1,6]$ code.

 Using the generalized tensor product construction, the code $\cC(\ell,H_1,\mathbf{1}_\ell, H_2)$
is a $(3,4)$-LRC with length $6\ell$, dimension $3\ell-2$ and minimum distance 6, for all $\ell \geq 2$.

\medskip

\underline{$\ell=\sigma+2$, $d=12$}

In the last case we need a structural property of the quaternary $[6,3,4]$ MDS code. 

\begin{prop}
If we pick two codewords of weight 6 in the $[6,3,4]$ MDS code that are not scalar multiple of each other, then the sum of them has weight~4.
\label{prop:MDS634}
\end{prop}

\begin{proof}
It is know that the $[6,3,4]$ MDS code is unique~\cite{Alderson06}. We just need to pick any $3\times 6$ matrix that generates the $[6,3,4]$ MDS code, and verify the proposition.
\end{proof}

Recall that the weight distribution of the $[6,3,4]$ MDS code is
$$
A_0=1,\ A_1=A_2=A_3=0, \ A_4=45,\ A_5=0,\ A_6=18.
$$
If a codeword has the nonzero code symbols spread over three or more repair groups, then the Hamming weight of the codeword is at least 12. Hence, we only need to consider codewords supported within two repair groups. When a nonzero codeword has support inside the union of two repair groups, then all code symbols in the two repair groups must be nonzero. We have to exclude the possibility of $4+4$, $4+6$, or $6+4$ nonzero symbols. 

Since the quaternary  $[6,3,4]$ code is unique there is no loss in generality if we take
$$
H_{3\times 6}= \begin{bmatrix}
1&0&0&1&1&1 \\
0&1&0&1&\alpha&\alpha^2 \\
0&0&1&1&\alpha^2&\alpha
\end{bmatrix}
$$
as a parity-check matrix. Since the repair groups must be disjoint (by Prop.~\ref{prop:disjoint_support}), we can write down a parity-check matrix of a $(3,4)$-LRC with distance 12 in the form,
$$ H = \left[
\begin{array}{cccc}
H_{3\times 6} &&&\\
&H_{3\times 6} &&\\
&&H_{3\times 6} &\\
&&&H_{3\times 6} \\ \hline
M_1 & M_2 & M_3 & M_4
\end{array} \right]
$$
where $M_j$ is a $5\times 6$ matrix whose columns 1, 2 and 3 are zero vectors,
$$
M_j = \begin{bmatrix}
0&0&0&u_1&v_1&w_1 \\
0&0&0&u_2&v_2&w_2 \\
0&0&0&u_3&v_3&w_3 \\
0&0&0&u_4&v_4&w_4 \\
0&0&0&u_5&v_5&w_5 
\end{bmatrix}.
$$

Suppose $\cC$ is an optimal LRC in this case consisting of $\ell$ repair groups. For $j=1,2,\ldots \ell$, let $\mathbf{u}_j$, $\mathbf{v}_j$ and $\mathbf{w}_j$ be the three columns vectors in the submatrix $M_j$. As the there is no codeword with Hamming weight 6 or less, the three vectors $\mathbf{u}_j$, $\mathbf{v}_j$ and $\mathbf{w}_j$ are linearly independent.

Let $\mathcal{L}$ be the null space of $H_{3\times 4}$, and for $a=4$ and $a=6$, define
$$
\mathcal{U}_a := \{\mathbf{c}\in \mathcal{L}:\,  wt_H(\mathbf{c})=a\}_{\{4,5,6\}}.
$$
The set consists of the restriction of the codewords with weight $a$ on the last three components.

For $j=1,2,\ldots, \ell$, let $V_j$ be the vector subspace spanned by $\mathbf{u}_j$, $\mathbf{v}_j$ and $\mathbf{w}_j$, and let
$$
\mathcal{A}_j^{(a)} := \{\phi \mathbf{u}_j+\theta \mathbf{v}_j+\tau\mathbf{w}_j:\, (\phi,\theta,\tau)\in\mathcal{U}_a\},
$$
for $a=4, 6$. The set $\mathcal{A}_j^{(a)}$ is closed under scalar multiplication, but need not be closed under addition. The set $\mathcal{A}_j^{(6)}$ can be interpreted as the union of 6 one-dimensional subspaces in $GF(4)^5$, excluding the origin.

Since $\dim(V_j)=3$ for all $j$ and $V_j$ is living in a vector space of dimension 5, two vector subspaces $V_j$ and $V_m$ associated to two repair groups $j$ and $m$ (with $j\neq m$) must intersect in a subspace of dimension at least 1, because
$$
\dim(V_j \cap V_m)  = \dim(V_j)+\dim(V_m) - \dim(V_j + V_m) \geq 3 + 3- 5= 1.
$$
As the minimum distance is 12, any nonzero vector in $V_j\cap V_m$ must be in $A_j^{(6)}\cap A_m^{(6)}$. However, the intersection cannot contain two linearly independent vectors. We claim that $\dim(V_j\cap V_m)=1$ for all $j\neq m$.

Suppose on the contrary that
\begin{align*}
\mathbf{c}&=(x_1,y_1,z_1,\phi_1,\theta_1,\tau_1,\ x_2,y_2,z_2,\phi_2,\theta_2,\tau_2, \ldots) \\
\mathbf{c}' &=(x_1',y_1',z_1',\phi_1',\theta_1',\tau_1',\ x_2',y_2',z_2',\phi_2',\theta_2',\tau_2', \ldots)
\end{align*}
are two codewords whose nonzero components are all in the first two repair groups, and the two vectors
$$
\phi_1\mathbf{u}_1+\theta_1\mathbf{v}_1+\tau_1\mathbf{w}_1, \
\phi_1'\mathbf{u}_1+\theta_1'\mathbf{v}_1+\tau_1'\mathbf{w}_1
$$
are linearly independent vector in $V_1$. Then 
$$
(x_1,y_1,z_1,\phi_1,\theta_1,\tau_1) \text{ and }
(x_1',y_1',z_1',\phi_1',\theta_1',\tau_1')
$$
are two codewords of weight 6 in the $[6,3,4]$ code that are not scalar multiple of each other, and likewise,
$$
(x_2,y_2,z_2,\phi_2,\theta_2,\tau_2) \text{ and }
(x_2',y_2',z_2',\phi_2',\theta_2',\tau_2')
$$
are another pair of codewords of weight 6 that are not scalar multiple of each other. Then by Prop.~\ref{prop:MDS634}, the sum of $\mathbf{c}$ and $\mathbf{c}'$ is a codeword of weight 8 in the LRC, contradicting the assumption on minimum distance. This finish the proof of the claim.

We have thus proved that for any two repair groups, say repair groups $j$ and $m$, the intersection of the respective vector subspaces must intersect in a one-dimensional subspace spanned by a vector in $\mathcal{A}_j^{(6)}\cap\mathcal{A}_m^{(6)}$.

We have the following upper bound on the number of repair groups.

\begin{prop}
An optimal $(3,4)$-LRC with minimum distance 12 has at most 21 repair groups.
\end{prop}

\begin{proof}
Suppose $\cC$ is an optimal $(3,4)$-LRC with minimum distance 12 and $\ell$ disjoint repair groups. The  vectors in $\cup_{j=1}^\ell \mathcal{A}_j^{(4)}$ must be distinct. If not, we will have a codeword of weight 8.  
Since the ambient vector space contains $4^5=1024$ vectors, the number of repair groups is no larger than $(1024-1)/45 = 22.733$. Hence, there are no more than 22 repair groups. We next show that $\ell=22$ is impossible.

Suppose there are $\ell=22$ repair groups. The union of $\mathcal{A}_j^{(4)}$, for $j=1,2,\ldots, 22$, occupy 990 nonzero vectors in $GF(4)^5$. The sets $\mathcal{A}_j^{(6)}$ belong to a set of size $1023-990=33$. However, for any two repair groups, say repair groups $j$ and $m$, the total number of vectors in $\mathcal{A}_j^{(6)} \cup \mathcal{A}_m^{(6)}$ is  $18+18-3 = 33$. There is no room for the third repair group. The number of repair groups is at most 2. This certainly contradicts the assumption that there are 22 repair groups.

\end{proof}

There is a construction for $\ell=17$ repair groups. In this construction we assume that there is a vector that is contained in $\mathcal{A}_j^{(6)}$ for all $j=1,2,\ldots, \ell$. As $\mathcal{A}_j^{(6)}$ is closed under scalar multiplication, there is a one-dimensional subspace, say $W$, in $GF(4)^5$, such that $W$ is contained in $V_j$ for all $j$. Consider the quotient fields $GF(4)^5/W$. For each $j$, the quotient fields $V_j/W$ for $j=1,2,\ldots \ell$, only intersect at the zero vector $\{\mathbf{0}\}$. Indeed there is a spread of size 17. We can use this spread to form an LRC consisting of 17 repair groups.

The code parameters are
\begin{equation}
n=6\ell, \ k =3\ell-5, \ r=3,\ \delta =4,\ d=12\ (3 \leq \ell\leq 17 ?).
\label{LRC:r=3_delta=4_large2}
\end{equation}

It is not known whether $(3,4)$-LRC of minimum distance 12 exists or not when then number of repair groups is between 18 and 21.

\section{Concluding Remarks}
\label{sec:conclusion}

In this paper we classify quaternary $(r,\delta)$ LRC with $\delta>2$. With three exceptions, all possible combinations of code parameters are found, and for each combination of code parameters we provide an explicit code construction. 
The result of the classification is shown in Table.~\ref{table:code_parameters}. The codes can be divided into two groups. The first group has locality $r=1$. The code length and minimum distance can get arbitrarily large in this group. In the second group, the locality parameter $r$ is larger than or equal to~2, and the minimum distance is bounded between 3 and~12. For minimum distance $d=3$ or $d=4$, we have a sequence of optimal LRCs for each combination of $(r,\delta)$. For minimum distance $d>4$, we have infinitely many optimal LRCs of increasing lengths only when $r=3$, $\delta\in\{3,4\}$ and $k>r+1$. In the remaining combinations of $r$ and $\delta$, there are either no or finitely many LRCs that are Singleton-optimal.

The first unsettled case in the classification is $(3,3)$-LRC with $\ell=\sigma+2$, where $\ell$ is the number of repair groups and $\sigma=\lfloor k/r \rfloor-1$.  The largest possible minimum distance is 10 by the Singleton-type bound. Nevertheless, due to the special weight distribution of the quaternary $[5,3]$ MDS code, we conjecture that minimum distance 10 is not achievable. The second open case is $(2,3)$-LRC with $\ell=\sigma+3$ repair groups and minimum distance~12. We give a construction for 5 repair groups, but it is now known whether longer code exists. The last unsettled case is $(r,\delta)=(3,4)$ and $d=12$. We prove that the repair groups are disjoint and the there are at most 21 repair groups. We can construct a $(3,4)$-LRC with 17 repair groups and minimum distance 12. Whether we can extend it to 18 to 21 repair groups is an open question.


\appendices

\section{Proof of Propositions~\ref{prop:chain} and ~\ref{prop:residue_code}}
\label{app:A}

\begin{proof}[Proof of Prop.~\ref{prop:chain}]
Suppose $\cC$ is an $(r,\delta)$-LRC with dimension $k$ and $(\cD_1,\ldots, \cD_M)$ is a maximal chain of subcodes in~$\mathscr{D}$. The dimension of the direct sum of $\cD_1$ to $\cD_M$ satisfies
\begin{equation}
n-k \geq \dim(\cD_1\oplus\cdots\oplus\cD_M) \geq M(\delta-1).
\label{eq:chain_dim}
\end{equation}
Let $n'$ be the size of the support of $\cD_1\oplus\cdots\oplus\cD_M$. If $n'=n$, the inequality in \eqref{eq:chain_proof} follows from \eqref{eq:chain_dim} directly. We can thus assume without loss of generality that $n'<n$. By permuting the code symbols if necessary, we may assume that 
$$\supp(\cD_1\oplus\cdots\oplus\cD_M) = \{1,2,\ldots, n'\}, 
$$
i.e., the last $n-n'$ components in the direct sum $\cD_1\oplus \cdots \oplus \cD_M$ are all zero. 

Let $i$ be an index in $\{n'+1, n'+2,\ldots, n\}$ and $\cD'$ be a subcode in $\mathscr{D}$ whose support contains $i$ (such a subcode $\cD'$ exists by Condition (a')). We claim that the index set
$$
\mathcal{X} = \supp(\cD') \setminus \supp(\cD_1\oplus\cdots\oplus\cD_M)
$$
has size strictly less than $\delta-1$. Otherwise, if $|\mathcal{X}| \geq \delta -1$, by Condition (c'), the restriction of $\cD'$ on $\mathcal{X}$ has dimension larger than or equal to $\delta - 1$. Then $\cD_1, \cD_2,\ldots, \cD_M, \cD'$ is a chain that contains $\cD_1, \cD_2, \ldots, \cD_M$. This contradicts the maximality of $\cD_1,\ldots, \cD_M$.

Therefore, by Condition (c') in the definition of LRC, we get $\dim(\cD'_{\mathcal{X}})=|\mathcal{X}|$. The columns of the corresponding local parity-check matrix indexed by $\mathcal{X}$ are linearly independent.  We can thus find a vector in $\cD'$ whose support contains $i$ but is confined in $\{1,2,\ldots, n'\}\cup\{i\}$. We call this vector $\mathbf{v}_i$, which is a codeword in $\cC^\perp$.

By repeating the above argument for each $i\in\{n'+1,n'+2,\ldots, n\}$, we can find $n-n'$ linearly independent codewords $\mathbf{v}_{n'+1}$, $\mathbf{v}_{n'+2}, \ldots, \mathbf{v}_n$ in $\cC^\perp$ that do not lie in $\cD_1\oplus \cdots \oplus \cD_M$. Since the dimension of $\cC^\perp$ is $n-k$, we get 
$$
n-k \geq M(\delta-1)+(n-n'),
$$
and this implies 
\begin{equation}
|\supp(\cD_1\oplus\cdots\oplus\cD_M)|-k \geq M(\delta - 1).
\label{eq:chain_proof}
\end{equation}

By Condition (b') in the definition of LRC, we have
$$
|\supp(\cD_1\oplus\cdots\oplus\cD_M)| \leq M(r+\delta - 1),
$$
as each repair group has size no larger than $r+\delta -1$. Combining with \eqref{eq:chain_proof}, we obtain
$$
 M(\delta-1)+k \leq |\supp(\cD_1\oplus\cdots\oplus\cD_M)| \leq  M(r+\delta-1),
$$
which can then be simplified to $M\geq \frac{k}{r}$. 
\end{proof}

\begin{proof}[Proof of Prop.~\ref{prop:residue_code}]
 Suppose $\cD$ is a subcode in $\cC^\perp$ that satisfies Conditions (b') and (c'). We can start with $\cD$ and find a maximal chain $\cD=\cD_1,\ldots, \cD_M$ with $\cD_1=\cD$. This can be done in a greedy manner. Given $\cD = \cD_1$, we can search for another subcode in $\mathscr{D}$ such that $\dim(\cD_1\oplus \cD_2)$ exceeds $\dim(\cD_1)$ by at least $\delta-1$. Since $\dim(\cC)$ is assumed to be strictly larger than $r$, we have $\lceil k/r \rceil \geq 2$ such a subcode must exist. We let $\cD_2$ be  one of such subcodes. We can continue similarly and create a chain that is as long as possible. Prop.~\ref{prop:chain} guarantees that the length of the resulting chain is larger than or equal to $\lceil k/r \rceil$.

Consider the first $m:=\lceil k/r \rceil-1$ subcodes $\cD_1,\ldots,\cD_m$ in this chain of subcodes. Let $\cX$ be the support of $\cD_1\oplus\cdots\oplus\cD_m$. The size of $\cX$ is at most $m(r+\delta-1)$, and the dimension of $\cD_1\oplus\cdots\oplus\cD_m$ is at least $m(\delta-1)$. 
We puncture the dual code $\cC^\perp$ by removing the code symbols with indices in $\cX$. The punctured code $(\cC^{\perp})_{\cX^c}$ has length no less than $n-m(r+\delta-1)$ and dimension no more than $n-k-m(\delta-1)$. This punctured code does not contain all vectors of length $|\cX^c|$, because
\begin{align}
 n-m(r+\delta-1) - [n-k-m(\delta-1)] 
&= k - mr \notag \\
&= k - \big\lfloor \frac{k-1}{r} \big\rfloor r \label{eq:trick} \\
&\geq 1. \notag 
\end{align}
In \eqref{eq:trick} we use an elementary fact that $m=\lceil k/r \rceil-1=\big\lfloor \frac{k-1}{r} \big\rfloor$. Hence, $(\cC^\perp)_{\cX^c}$ is not a trivial code.

The code $(\cC^{\perp})_{X^c}$ is the dual code of the shortened code of $\cC$ consisting of codewords that are all zero in $\cX$. This shortened code contains some nonzero codewords, as $(\cC^\perp)_{\cX^c}$ is not a trivial code. By the Singleton bound, the minimum distance of this shortened code is upper bounded by
\begin{align*}
 \dim( (\cC^\perp)_{\cX^c})+1 \leq n - k-  \big( \big\lceil \frac{k}{r} \rceil - 1\big)(\delta - 1) + 1 ,
\end{align*}
which is identical to the Singleton-like bound in~\eqref{eq:Singleton_like}. Because $\cC$ is assumed to be Singleton-optimal and shortening does not decrease minimum distance, we have equality 
$$
\dim(\cD_1\oplus\cD_2\oplus\cdots\oplus\cD_i) - 
\dim(\cD_1\oplus\cD_2\oplus\cdots\oplus\cD_{i-1}) = \delta - 1
$$
for $i=1,2,\ldots, m$. In particular, the dimension of $\cD=\cD_1$ is exactly equal to $\delta-1$.

For the second statement in the proposition, we let $\cR$ be the support of a repair group and let $\cD$ be the corresponding subcodes in $\cC^\perp$, i.e., $\cD$ is the shortened subcode of $\cC^\perp$ such that $\cD_\cR$ and $\cC_\cR$ are dual to each other. From the previous paragraph we have $\dim(\cD)=\delta-1$, and by the nullity theorem of linear algebra, $\dim(\cC_\cR) = |\cR| - \dim(\cD) = |\cR|-\delta+1$. By applying the Singleton bound again, we obtain $d_{\mathrm{min}}(\cC_\cR) \leq |\cR| - \dim(\cC_\cR) = \delta-1$.
\end{proof}

\section{Proof of Proposition~\ref{prop:GTP}}

We first show that the rows of the matrix $H$ are linearly independent. For $i=1,2,\ldots, n_1,$ let $\bc_i$ denote a vector of length $\delta-1$, and for $j=1,2,\ldots, \mu$, and let $\bd_j$ denote a vector of length~$\nu$. Consider the concatenated vector
\begin{equation}
(\bc_1, \bc_2,\ldots, \bc_{n_1}, \bd_1, \bd_2,\ldots, \bd_\mu).
\label{eq:concatenated_vector}
\end{equation}
Suppose this vector is in the left null-space of $H$. For $\ell=1,2,\ldots, n_1$, the $\ell$-th block of columns in $H$ yields the following equality,
$$
\bc_1 B_1 + (a_{1\ell} \bd_1 +  a_{2 \ell} \bd_2 + \cdots + a_{\mu \ell} \bd_\mu) B_2  = \mathbf{0}.
$$
Using the assumption that the rows of $B_1$ and $B_2$ are linearly independent, we conclude that the components in $\bc_1$ and $a_{1\ell} \bd_1 +  a_{2\ell} \bd_2 + \cdots + a_{\mu \ell} \bd_\mu$ are all zero.

For $j=1,2,\ldots \mu$, let the $j$-th component in vector $\bd_\ell$ be denoted by $d_\ell^j$. Because $a_{1\ell} \bd_1 +  a_{2\ell} \bd_2 + \cdots + a_{\mu \ell} \bd_\mu$ is a zero vector for $\ell=1,2,\ldots, n_1$, we have
$$
(d_1^j, d_2^j, \ldots , d_\mu^j) \cdot A_2 = \mathbf{0} 
$$
for all $j$. Since the rows of $A_2$ are linearly independent, we get $d_\ell^j=0$ for $j=1,2,\ldots, n_1$ and $\ell=1,2,\ldots, \mu$. This proves that if the concatenated vector in~\eqref{eq:concatenated_vector} is in the left null space of $H$, then it must be the zero vector. 

Hence, the parity-check matrix $H$ defines a linear code with dimension
$$
  n_1(r+\delta-1) - [n_1(\delta-1) + \mu\nu] = n_1r - \mu \nu.
$$

We next calculate the minimum distance of the linear code defined by $H$. Write a codeword $\bu$ as a concatenated vector
\begin{equation}
\bu = \begin{bmatrix} \bu_1 & \bu_2 & \cdots & \bu_{n_1} \end{bmatrix},
\label{eq:concatenated_codeword}
\end{equation}
where $\bu_\ell$ is a row vector of length $r+\delta-1$, for $\ell=1,2,\ldots,n_1$. The defining condition $H \bu^T$ can be broken into
\begin{align}
B_1 \bu_\ell^T &= \mathbf{0}, \ \text{for } \ell=1,2,\ldots, n_1, \label{eq:proof1} \\
\sum_{\ell=1}^{n_1} a_{j\ell} B_2 \bu_\ell^T &= \mathbf{0}, \ \text{for } j=1,2,\ldots, \mu. \label{eq:proof2}
\end{align}

Suppose that the codeword $\bu$ is nonzero. We can find a non-empty index set $\mathcal{L}$ such that $\bu_{\ell} \neq \mathbf{0}$ when and only when $\ell \in \mathcal{L}$. We consider two cases. Firstly, suppose that $r=\nu$, i.e., the matrix $B:=\begin{bmatrix} B_1 \\ B_2 \end{bmatrix}$ is a nonsingular square matrix. Because $B_1 \mathbf{u}_\ell^T= \mathbf{0}$, the Hamming weight of $\bu_\ell$ is larger than or equal to $\delta$. Since $B_1 \bu_\ell^T=\mathbf{0}$ and $B$ is nonsingular, the product $ B_2 \bu_{\ell}^T$ must be nonzero for all $\ell \in \mathcal{L}$. The hypothesis that $A_2$ is a parity-check matrix of an $[n_1, n_1-\mu]$ MDS code, together with the condition in \eqref{eq:proof2}, imply that $|\mathcal{L}|\geq \mu+1$. Hence, the minimum distance of the code obtained by Construction~1 is larger than or equal to $\delta(\mu+1)$. The proof of the first case is completed by noting that there is indeed a codeword $\mu$ with Hamming weight exactly equal to $\delta(\mu+1)$.

 Next, we assume that $r > \nu$. Similar to the first case described in the previous paragraph, we can find codewords with Hamming weight $\delta(\mu+1)$. On the other hand, because $\begin{bmatrix} B_1 \\ B_2 \end{bmatrix}$ is a parity-check matrix of an $[r+\delta-1, r-\nu,\delta+\nu]$ MDS code, we can find codewords with Hamming weight equal to $\delta+\nu$. Indeed we can set one of the $\bu_\ell$'s to be a vector with Hamming weight precisely $\delta+\nu$, and the rest to the zero vector. This shows that the minimum distance of the code obtained by Construction~1 is no less than the minimum of $\delta(\mu+1)$ and $\delta+\nu$.



\end{document}